%% file: main.tex
\documentclass[a4paper,USenglish,cleveref,authver]{lipics-v2019}

\usepackage[noend]{algpseudocode}
\usepackage{algorithm}
\usepackage{mathtools,mathrsfs}
\usepackage{subcaption}
\usepackage{hyperref}
\usepackage{cite}
\usepackage{times}
\usepackage{tikz}
\nolinenumbers
\hideLIPIcs

\newcommand{\blank}[1]{\hspace*{#1}}

\graphicspath{{./figs/}}

\newcommand{\cmnt}[1]{}
\newcommand{\algoref}[1]{{Algorithm \ref{alg:#1}}}
\newcommand{\defref}[1]{Definition~\ref{def:#1}}
\newcommand{\obsref}[1]{Observation~\ref{obs:#1}}

\newcommand{\lemref}[1]{Lemma~\ref{lem:#1}}
\newcommand{\figref}[1]{Figure~\ref{fig:#1}}

\newcommand{\theoremref}[1]{Theorem~\ref{theorem:#1}}

\newcommand{\fnonpartionable}{f_{opt}}
\newcommand{\fmajority}{f_{maj}}
\newcommand{\fmandm}{f_{\text{m\&m}}}
\newcommand{\fG}{f_G}
\newcommand{\tL}{f_L}
\newcommand{\fc}{f_{\text{cluster}}}
\newcommand{\regsno}{\rho}
\newcommand{\readno}{\sigma}
\newcommand{\exec}{\alpha}
\newcommand{\bb}{{\sc{SDC}}}

\newcommand{\sharedmemcomm}[3][]{
\ifthenelse{\equal{#1}{not}}{#2 \not\leftrightarrow #3}{#2 \leftrightarrow #3}}
\newcommand{\communicateset}[1]{\overset{\rightarrow}{#1}}

\newcommand{\canreadcomm}[3][]{
\ifthenelse{\equal{#1}{not}}
{#2 \not\rightarrow #3}{#2 \rightarrow #3}}
\newcommand{\bothcanreadcomm}[3][]{
\ifthenelse{\equal{#1}{not}}
{#2 \not\leftrightarrow #3}{#2 \leftrightarrow #3}}

\newtheorem{observation}{Observation}

\algnewcommand\algorithmicforeach{\textbf{for each}}
\algdef{S}[FOR]{ForEach}[1]{\algorithmicforeach\ #1\ \algorithmicdo}

\title{Optimal Resilience in Systems that Mix\\Shared Memory and Message Passing}
\titlerunning{Resilience of Systems that Mix Shared Memory and Message Passing}
\author{Hagit Attiya}{Department of Computer Science, Technion, Israel}{hagit@cs.technion.ac.il}{0000-0002-8017-6457}{}
\author{Sweta Kumari}{Department of Computer Science, Technion, Israel}{sweta@cs.technion.ac.il}{}{}
\author{Noa Schiller}{Department of Computer Science, Technion, Israel}{noa.schiller@cs.technion.ac.il}{}{}
\Copyright{Hagit Attiya, Sweta Kumari and Noa Schiller}
\authorrunning{Attiya, Kumari and Schiller}


\begin{document}
\maketitle
\input{abstract}
\input{introduction}
\input{f-partitionable}
\input{swmr}
\input{tasks}
\input{mm}
\input{cluster-based}

\input{discussion}

\bibliographystyle{plainurl}
\bibliography{citations}

\end{document}

%% file: abstract.tex
\begin{abstract}
We investigate the minimal number of failures that can \emph{partition}
a system where processes communicate both through shared memory
and by message passing.
We prove that this number precisely captures the resilience that
can be achieved by algorithms that
implement a variety of shared objects,
like registers and atomic snapshots,
and solve common tasks, like randomized consensus,
approximate agreement and renaming.
This has implications
for the \emph{m\&m-model} of~\cite{Aguilera:PODC:2018}
and for the hybrid, cluster-based model
of~\cite{ImbsR13,RaynalCao:Clusterbased:ICDCS:2019}.
\end{abstract}

%% file: introduction.tex
\section{Introduction}

Some distributed systems combine more than one mode of communication
among processes, allowing them both to send messages among themselves
and to access shared memory.
Examples include recent technologies such as
\emph{remote direct memory access}
(\emph{RDMA})~\cite{InfiniBand, iWARP, RDMA},
\emph{disaggregated memory}~\cite{Lim+:ISCA:2009},
and \emph{Gen-Z}~\cite{Gen-Z}.
In these technologies, the crash of a process does not prevent
access to its shared memory by other processes.
Under these technologies,
it is infeasible to share memory among a large set of processes,
so memories are shared by smaller, strict subsets of processes.

Systems mixing shared memory and message passing offer a major
opportunity since information stored in shared variables remains
available even after the failure of the process who stored it.
Mixed systems are expected to withstand more process
failures than pure message-passing systems, as captured by
the \emph{resilience} of a problem---the maximal
number of failures that an algorithm solving this problem can tolerate.
This is particularly the case in an \emph{asynchronous} system.
At one extreme, when all processes can access the same shared memory,
many problems can be solved even when all processes but one fail.
Such \emph{wait-free} algorithms exist for implementing shared objects
and solving tasks like randomized consensus, approximate agreement
and renaming.
At the other extreme, when processes only communicate by message passing,
the same problems require that at least a majority of processes
do not fail~\cite{AttiyaBD1995,AttiyaBDPR90,BrachaT1985}.
Thus, typically, shared-memory systems are $(n-1)$-resilient, and
pure message-passing systems are $\lfloor{(n-1)/2}\rfloor$-resilient,
where $n$ is the number of processes.

The resilience in systems that mix shared memory
and message passing falls in the intermediate range,
between $\lfloor{(n-1)/2}\rfloor$ and $n-1$.
It is, however, challenging to solve specific problems with the
best-possible resilience in a particular system organization:
the algorithm has to coordinate between non-disjoint sets of processes
that have access to different regions of the shared memory.
On the other hand, bounding the resilience requires to take
into account the fact that processes might be able to
communicate \emph{indirectly} through shared memory accesses
of third-party processes.

This paper explores the \emph{optimal} resilience in systems
that provide message-passing support between all pairs of processes,
and access to shared memory between subsets of processes.
We do this by studying the minimal number of failures that
can \emph{partition} the system, depending on its structure,
i.e., how processes share memory with each other.
We show that the partitioning number exactly characterizes the
resilience, that is, a host of problems can be solved in the presence
of $< f$ crash failures, if and only if
$f$ is the minimal number of failures that partition the system.

A key step is to focus on the implementation of a \emph{single-writer
multi-reader register} shared among all processes,
in the presence of $f$ crash failures.
A read or a write operation takes $O(1)$ round-trips,
and requires $O(n)$ messages.
Armed with this implementation,
well-known shared-memory algorithms can be employed to
implement other shared objects,
like \emph{multi-writer multi-reader registers}
and \emph{atomic snapshots},
or to solve fundamental problems,
such as \emph{randomized consensus}, \emph{approximate agreement}
and \emph{renaming}.
Because the register implementation is efficient,
these algorithms inherit the good efficiency of the best-known
shared-memory algorithm for each of these problems.

Going through a register simulation,
instead of solving consensus, approximate agreement
or renaming from scratch,
does not deteriorate their resilience.
One of our key contributions is to show that the resilience achieved
in this way is optimal,
by proving that these problems cannot be solved in the
presence of $f$ crash failures,
if $f$ failures can partition the system.

We consider memories with access restrictions and model mixed systems
by stating which processes can read from or write to each memory.
(Note that every pair of processes can communicate using messages.)
Based on this concept, we define $\fnonpartionable$ to be
the largest number of failures that do not partition the system.
We prove that $f$-resilient registers and snapshot implementations,
and $f$-resilient solutions to randomized consensus,
approximate agreement and renaming,
exist if and only if $f \leq \fnonpartionable$.

One example of a mixed model is the \emph{message-and-memory}
model~\cite{Aguilera:PODC:2018}, in short, the \emph{m\&m} model.
In the \emph{general} m\&m model~\cite{Aguilera:PODC:2018},
the shared-memory connections are defined by
(not necessarily disjoint) subsets of processes,
where each subset of processes share a memory.
Most of their results, however, are for the \emph{uniform} m\&m model,
where shared-memory connections can be induced by an undirected graph,
whose vertices are the processes.
Each process has an associated shared memory that can be accessed
by all its neighbors in the \emph{shared-memory graph}
(see Section~\ref{section:mm}).
They present bounds on the resilience for solving randomized consensus in the uniform model.
Their algorithm is based on Ben-Or's exponential algorithm for
the pure message-passing model~\cite{BenOr1983}.
The algorithm terminates if the nonfaulty processes and their
neighbors (in the shared-memory graph) are a majority of the processes.
They also prove an upper bound on the number of failures a randomized
consensus algorithm can tolerate in the uniform m\&m model.
We show that in the uniform m\&m model, this bound is equivalent to
the partitioning bound ($\fnonpartionable$) proved in our paper
(\theoremref{cut&partition-eq} in Section~\ref{section:mm}).
We further show that this bound \emph{does not match the resilience of
their algorithm},
whose resilience is strictly smaller than $\fnonpartionable$,
for some shared-memory graphs.

In the special case where the shared memory \emph{has no access restrictions},
our model is dual to the general m\&m model,
i.e., it captures the same systems as the general m\&m model.
However, rather than listing which processes can access a memory,
we consider the flipped view: we consider for each process,
the memories it can access.
We believe this makes it easier to obtain some extensions,
for example, for memories with access restrictions.

Hadzilacos, Hu and Toueg~\cite{Hadzilacos+:M&M:OPODIS:2019} present
an implementation of a SWMR register in the general m\&m model.
The resilience of their algorithm is shown to match the maximum
resilience of an SWMR register implementation in the m\&m model.
Our results for register implementations are adaptations of their results.
For the general m\&m model specified by the set of process subsets $L$,
they define a parameter $\tL$ and show that it is
the maximum number of failures tolerated by an algorithm implementing
a SWMR register~\cite{Hadzilacos+:M&M:OPODIS:2019} 
or solving randomized consensus~\cite{HadzilacosHT2020arxiv}.
For memories without access restrictions, $\tL$ is equal to $\fnonpartionable$.
Their randomized consensus algorithm is based on the simple
algorithm of~\cite{AspnesH1990} and inherits its exponential expected
step complexity.

Another example of a model that mixes shared memory and
message passing is the hybrid model of~\cite{ImbsR13,
RaynalCao:Clusterbased:ICDCS:2019}.
In this model, which we call \emph{cluster-based},
processes are partitioned into disjoint \emph{clusters},
each with an associated shared memory; all processes in the cluster
(and only them) can read from and write to this shared memory.
Two randomized consensus algorithms are presented for
the cluster-based model~\cite{RaynalCao:Clusterbased:ICDCS:2019}.
Their resilience is stated as an operational property of executions:
the algorithm terminates if the clusters of responsive processes
contain a majority of the processes.
We prove (\lemref{fmajorityleqfnonpartion} in Section~\ref{sec:cluster})
that the optimal resilience we state in a closed form
for the cluster-based model is equal to their operational property.

Our model is general and captures all these models within a single
framework, by precisely specifying the shared-memory layout.
The tight bounds in this general model provide the exact resilience
of any system that mix shared memory and message passing.


%% file: f-partitionable.tex
\section{Modelling Systems that Mix Shared Memory and Message Passing} \label{sys-model}

We consider $n$ asynchronous processes $p_1, \dots, p_n$,
which communicate with each other by sending and receiving messages,
over a complete communication network of asynchronous reliable links.
In addition, there are $m$ shared memories $M=\{\mu_1,...,\mu_m\}$,
which can be accessed by subsets of the processes.
A memory $\mu\in M$ has access restrictions,
where $R_\mu$ denotes all the processes that can read from
the memory and $W_\mu$ denotes all the processes that
can write to the memory.
The set of memories a process $p$ can read from is denoted $R_p$,
i.e., $R_p=\{\mu\in{M}:p\in{R_{\mu}}\}$.
The set of memories $p$ can write to is denoted $W_p$,
i.e., $W_p=\{\mu\in{M}:p\in{W_{\mu}}\}$.
We assume the network allows nodes to send the same message to
all nodes; message delivery is FIFO.
A process $p$ can \emph{crash}, in which case it stops taking steps;
messages sent by a crashed process may not be delivered at their recipients.
We assume that the shared memory does not fail, as done in
prior work~\cite{Aguilera:PODC:2018,Hadzilacos+:M&M:OPODIS:2019,ImbsR13,
RaynalCao:Clusterbased:ICDCS:2019}.

\begin{definition}\label{def:sharedmemcomm}
If there is a shared memory $\mu\in{M}$ that $p$ can read from
and $q$ can write to, then we 
denote $\canreadcomm{p}{q}$.
If $\canreadcomm{p}{q}$ and $\canreadcomm{q}{p}$,
then we denote $\bothcanreadcomm{p}{q}$.
\end{definition}

Since a process can read what it writes to its local memory, this relation
is \emph{reflexive}, i.e., for every process $p$, $\canreadcomm{p}{p}$.

\begin{definition}\label{def:sharedmemcommset}
Let $P$ and $Q$ be two sets of processes.
Denote $\canreadcomm{P}{Q}$ if some process $p \in P$ can read
what a process $q \in Q$ writes, i.e., $\canreadcomm{p}{q}$.
If $\canreadcomm{P}{Q}$ and $\canreadcomm{Q}{P}$,
then we denote $\bothcanreadcomm{P}{Q}$.
\end{definition}

\begin{definition}\label{def:fpartition}\label{def:fnonpartion}
A system is \emph{$f$-partitionable} if there are two sets of processes
$P$ and $Q$, both of size $n-f$, such that $\bothcanreadcomm[not]{P}{Q}$.
Namely, the failure of $f$ processes can partition (disconnect) two
sets of $n-f$ processes.\\
Denote by $\fnonpartionable$ the largest integer $f$ such that
$\bothcanreadcomm{P}{Q}$,
for every pair of sets of processes $P$ and $Q$, each of size $n-f$.
\end{definition}

Clearly, a system is $f$-partitionable if and only if
$f>\fnonpartionable$.
Note that $\fnonpartionable\geq{\lfloor{(n-1)/2}\rfloor}$.
In the pure message-passing model, without shared memory,
$\canreadcomm{p}{q}$ if and only if $p=q$;
hence, $\fnonpartionable=\lfloor{(n-1)/2}\rfloor$.

The special case of shared memory \emph{without access restrictions}
is when for every memory $\mu\in M$, $R_{\mu} = W_{\mu}$,
and all processes that can read from a memory can also write to it.
In this case, the $\canreadcomm{}{}$ relation is \emph{symmetric},
i.e., for every pair of processes $p$ and $q$,
if $\canreadcomm{p}{q}$ then $\canreadcomm{q}{p}$.
Therefore, for every two processes $p$ and $q$, $\bothcanreadcomm{p}{q}$.
Later, we discuss two models without access restrictions,
the m\&m model (Section~\ref{section:mm})
and the cluster-based model (Section~\ref{sec:cluster}).

For a set of processes $P$, $\communicateset{P}$ are the processes
that some process in $P$ can read what they write to the shared memory,
i.e., $\communicateset{P} = \{q:\exists{p\in{P}},\canreadcomm{p}{q}\}$.

\begin{definition}\label{def:fmajority}  
$\fmajority$ is the largest integer $f$ such that for every set $P$ of
$n-f$ processes, $| {\communicateset{P}} | >\lfloor{n/2}\rfloor$.
That is, $\fmajority$ is the largest number of failures that still allows
the remaining (nonfaulty) processes
to communicate with a majority of the processes.
\end{definition}



\begin{lemma}\label{lem:fnonpartionleqfmajority}
$\fnonpartionable \leq \fmajority$.
\end{lemma}

\begin{proof}
Assume, by way of contradiction,
that there is set of $n-\fnonpartionable$ processes, $P$,
such that $| {\communicateset{P}} | \leq{\lfloor{n/2}\rfloor}$.
Note that $\fnonpartionable\neq{\lfloor{(n-1)/2}\rfloor}$,
otherwise $n-\fnonpartionable>\lfloor{n/2}\rfloor$ and hence,
$| {\communicateset{P}} | >{\lfloor{n/2}\rfloor}$, which is a contradiction.
Let $Q$ be a set of $n-\fnonpartionable$ processes
not in $\communicateset{P}$.
This set exists since $n-|\communicateset{P}| \geq\lfloor{n/2}\rfloor$
and $n-\fnonpartionable < n-\lfloor{(n-1)/2}\rfloor = {\lfloor{n/2}\rfloor} + 1$
implies that $n-\fnonpartionable\leq{\lfloor{n/2}\rfloor}$.
By definition of $\communicateset{P}$,
$\canreadcomm[not]{P}{Q}$,
contradicting the definition of $\fnonpartionable$.
\end{proof}

The converse direction does not necessarily hold,
as discussed for the m\&m model
and the cluster-based model. 

%% file: swmr.tex
\section{Implementing a Register in a Mixed System}
\label{section:register}

This section shows that a register can be implemented in the
presence of $f$ failures, if and only if the system is not $f$-partitionable,
that is, $f \leq \fnonpartionable$.
A \textit{single-writer multi-reader} (SWMR) register $R$
can be written by a single writer process $w$,
using a procedure {\sc{Write}},
and can be read by all processes $p_1, \dots, p_n$,
using a procedure {\sc{Read}}.
A register is \emph{atomic}~\cite{Lamport86} if any execution of
{\sc Read} and {\sc Write} operations can be
linearized~\cite{HerlihyW1990}.
This means that there is a total order of all completed operations
and some incomplete operations, that respects the real-time order
of non-overlapping operations, in which each {\sc Read} operation
returns the value of the last preceding {\sc Write} operation
(or the initial value of the register, if there is no such {\sc Write}).

\subsection{Implementing an Atomic Register in a Non-Partitionable System}
\label{swmr-impl}


The algorithm appears in \algoref{swmr}; for simplicity of presentation,
a process sends each message also to itself and responds with the appropriate response.
This is an adaptation of the register
implementation of~\cite{Hadzilacos+:M&M:OPODIS:2019} in the m\&m model.

All the message communication between the processes is done in msg\_exchange(),
where a process simply send a message and wait for $n-f$ acknowledgement.
This modular approach allows us to replace the communication pattern
according to the specific shared-memory layout.
For example, Section~\ref{sec:cluster} shows that in the cluster-based
model this communication pattern can be changed to wait for less than $n-f$ processes.

\begin{algorithm}
	\caption{Atomic SWMR register implementation ($w$ is the single writer).}	
	\label{alg:swmr}
	\begin{algorithmic}[1]
			\\ []{\bf Local Variables:}\\[]
			 {\textit{w-sqno}}: int, initially 0 \Comment{write sequence number}\\[]
			 {\textit{r-sqno}}: int, initially 0 \Comment{read sequence number}\\[]
			 {\textit{last-sqno}}: int, initially 0 \Comment{last write sequence number observed}\\[]
			  {\textit{counter}}:  int, initially 0
			 \Comment{number of replies/acks received so far}
			  \\ []{\bf Shared Variables:}  \\[]
			  \textbf{for each} process $p$ and memory $\mu\in{W_p}$: \\[]
			  \quad {\textit{R$_\mu$[p]}}: $\langle$int, int$\rangle$, initially $\langle0, v_0\rangle$
                    \Comment{writable by $p$ and readable by all processes that can read
                    from $\mu$, i.e., all the processes in $R_\mu$}\\[]
		 \hrulefill
			\makeatletter\setcounter{ALG@line}{0}\makeatother
			\item[] \blank{-.7cm} {\sc{Write}}(\textit{v}) --- \textbf{Code for the writer $w$:}
			\State \textit{w-sqno} = \textit{w-sqno} + 1 \Comment{increment the write sequence number}\label{lin:proposer1}
			\State \textit{acks} =  {msg\_exchange}$\langle$W\textit{, w-sqno, v}$\rangle$\label{lin:proposer2}
			\State return
			\item[] \vspace{-0.2cm}
			\item[] \blank{-.7cm} \textbf{Code for any process \textit{p}:}
			\State Upon receipt of a $\langle$W/WB\textit{, sqno, v}$\rangle$ message
                    from process \textit{w/q}:
			\If {\textit{(sqno $>$ last-sqno)}} \label{lin:line5}
			\State \textit{last-sqno} = \textit{sqno}
			\ForEach {$\mu\in{W_p}$}
			    \Comment{write value and sequence number to every register $p$ can write to}
			\State \textit{R$_\mu$[p]} = \textit{$\langle$sqno, v}$\rangle$
			    \label{lin:line6}
			\EndFor
			\EndIf
			\State send $\langle$Ack-W/Ack-WB\textit{, sqno}$\rangle$ to process \textit{w/q}
			\item[] \vspace{-0.2cm}
			\item[] \blank{-.7cm} {\sc{Read}}() --- \textbf{Code for the reader $q$:}
			\State \textit{r-sqno} = \textit{r-sqno} + 1 \Comment{increment the read sequence number}\label{lin:rproposer1}
			\State \textit{set\_of\_tuples} =
                {msg\_exchange}$\langle$R\textit{, r-sqno, $\bot$}$\rangle$\label{lin:rproposer2}
			\State \textit{$\langle$seq, val}$\rangle$ = max\textit{(set\_of\_tuples)} \Comment{maximum \textit{$\langle$seq, val$\rangle$}} \label{lin:line12}
			\State \textit{acks} = {msg\_exchange}$\langle$WB\textit{, seq, val}$\rangle$\Comment{write back} \label{lin:line13}
			\State return \textit{val}
			\item[] \vspace{-0.2cm}
			\item[] \blank{-.7cm} \textbf{Code for any process \textit{p}:}
			\State Upon receipt of a $\langle$R\textit{, r-sqno, -}$\rangle$
                    message from process \textit{q}:
			\State \textit{$\langle$w-seq, w-val}$\rangle$ =
                max\{\textit{$\langle$seq, val}$\rangle$ : $\mu\in{R_p\cap{W_q}}$ and \textit{R$_\mu$[q] = $\langle$seq, val}$\rangle$\} \Comment{find \textit{val} with maximum \textit{seq}} \label{lin:line17}
			\State send $\langle$Ack-R\textit{, r-sqno, $\langle$w-seq, w-val}$\rangle$$\rangle$ to process \textit{q} \label{lin:rproposer5}
			\item[] \vspace{-0.2cm}
			\item[] \blank{-.7cm}
			\textbf{msg\_exchange\textit{$\langle$m, seq, val}$\rangle$: returns \textit{set of responses}} \label{lin:proposer16}
			\State send \textit{$\langle$m, seq, val}$\rangle$ to all processes
			\State \textit{responses = $\emptyset$}
			\Repeat
			\State wait to receive a message $m$ of the form $\langle$Ack\textit{-m, seq, -}$\rangle$
			\State \textit{counter} = \textit{counter} + 1
			\State \textit{responses} = \textit{responses} $\cup$ \{$m$\}
			\Until \textit{counter}$\geq n-f$\label{lin:proposer22}
			\State return\textit{(responses)}
		\end{algorithmic}
\end{algorithm}	

For each process $p$ and memory $\mu\in W_p$ there is a shared SWMR
register $R_\mu[p]$, writable by $p$ and readable by every process
that can read from $\mu$, i.e., every process in $R_\mu$.
For simplicity, we assume that each process $p$ has a shared memory
that it can read from and write to, i.e.,
a SWMR register that $p$ can both write to and read from.
This assumption can be made since this register can be
in the process's local memory.
In {\sc{Write($v$)}},
the writer $w$ increments its local write sequence number \textit{w-sqno}
and calls msg\_exchange().
This procedure sends a message of type W with value $v$
and \textit{w-sqno} to all processes.
On receiving a write message from $w$,
$p$ checks if the write value is more up-to-date than the last value it
has observed, by checking if \textit{w-sqno} is larger than \textit{last-sqno}.
If so, $p$ updates \textit{last-sqno} to be \textit{w-sqno} and
writes the value $v$ and sequence number \textit{w-sqno}
to all the registers it can write to.
When done, the process sends an acknowledgment to the writer $w$.
Once $w$ receives $n-f$ acknowledgments, it returns successfully.

In {\sc{Read}}, a reader process $q$ increments its local
read sequence number \textit{r-sqno} and calls msg\_exchange().
This procedure sends a message of type R and \textit{r-sqno} to all processes.
On receiving a read message from $q$, a process $p$ reads all the registers
it can read and finds the maximum sequence number and value stored in them
and sends this pair to the reader $q$.
Once $q$ receives $n-f$ acknowledgments,
it finds the value \textit{val} with maximum sequence number \textit{seq}
among the responses (i.e., it selects the most up-to-date value).
Then, $q$ calls  msg\_exchange(), with message type WB (write back)
and value \textit{val} and \textit{seq} to update other readers.
On receiving a write back message from $q$,
each process $p$ handles WB like W message,
checking if \textit{w-sqno} is larger than \textit{last-sqno}
and if so updating \textit{last-sqno} and all the registers it can write to.
When done, the process sends an acknowledgment to $q$.
Once $q$ receives $n-f$ acknowledgments,
it returns $val$ successfully.

The communication complexities of read and write operations are
dominated by the cost of a msg\_exchange(),
invoked once in a write and twice in a read.
This procedure takes one round-trip and $O(n)$ messages,
like the algorithm for the pure message-passing model~\cite{AttiyaBD1995}.
The number of shared SWMR registers depends on the shared-memory topology
and is ${\regsno = \sum_{\text{process $p$}}{|W_p|}}$,
as every process has a single register in each memory it can write to.
The number of accesses to the shared memory is
${\readno = \sum_{\text{process $p$}}{\sum_{\mu\in R_p}{|W_\mu|}}}$,
as every process reads all the registers it can read from.
Note that $\readno \leq n\regsno$,
since each process reads at most $\regsno$ registers.
Formally, since $p\in W_\mu \Leftrightarrow \mu \in W_p$,
$\regsno = \sum_{\text{process $p$}}{|W_p|} = \sum_{\mu\in M}{|W_\mu|}$.
Hence,
$$\readno = \sum_{\text{process $p$}}{\sum_{\mu\in R_p}{|W_\mu|}}
          \leq \sum_{\text{process $p$}}{\sum_{\mu\in M}{|W_\mu|}}
          = \sum_{\text{process $p$}}{\regsno} = n\regsno ~. $$


\subsection{Correctness of the Algorithm} \label{swmr-correct}

We argue that \algoref{swmr} is correct when the system is not $f$-partitionable.
The only statement that could prevent the completion of a {\sc Write}
or a {\sc Read} is waiting for $n-f$ responses (Line~\ref{lin:proposer22}).
Since at most $f$ processes may crash,
the wait statement eventually completes,
implying the next lemma.

\begin{lemma}\label{lem:rwcrash}
A {\sc Write} or {\sc Read} invoked by a process that does not crash completes.
\end{lemma}

\begin{lemma}\label{lem:msgexw}
Let $t_2$ be the largest sequence number returned in a
read msg\_exchange by reader $p_j$,
and assume that the msg\_exchange starts after the completion of
a write msg\_exchange,
either by the writer $w$ or in a write back by reader $p_i$,
with sequence number $t_1$, then, $t_1$ $\leq$ $t_2$.
\end{lemma}

\begin{proof}
Let $P$ be the set of processes that respond in the write msg\_exchange,
and let $P'$ be the set of processes that respond in
the read msg\_exchange.
By Line~\ref{lin:proposer22}, $|P|, |P'| \geq n-f$.
Since the system is not $f$-partitionable, $\bothcanreadcomm{P}{P'}$.
Thus, there are processes $p\in{P}$ and $p'\in{P'}$ such that
$\canreadcomm{p'}{p}$,
and there is a shared memory $\mu\in{W_p\cap{R_{p'}}}$
that $p$ can write to and $p'$ can read from.
	
Upon receiving a write message (W or WB) with sequence number $t_1$,
$p$ compares \textit{last-sqno} with $t_1$ (Line~\ref{lin:line5}).
If $t_1$ is greater than \textit{last-sqno},
then $p$ updates all the registers it can write to.
In particular, it updates $\textit{R}_\mu[p]$ with the value associated
with sequence number $t_1$ (Line~\ref{lin:line6}), before responding.
If $\textit{last-sqno}\geq t_1$,
by the update in lines~\ref{lin:line5}-\ref{lin:line6}
and the fact that the shared registers are SWMR registers,
$\textit{R}_\mu[p] = \langle\textit{last-sqno},-\rangle$.
Hence, $\textit{R}_\mu[p] \geq t_1$
when $p$ sends a write response $\langle \text{Ack-W}, t_1 \rangle$.

Upon receiving a read message from $p_j$,
$p'$ reads all the registers it can read, including those in $\mu$,
and selects the value associated with the maximum sequence number.
In particular, $p'$ reads $\textit{R}_\mu[p]$,
and since this read happens after $p$ sends the write response,
$p'$ reads a value larger than or equal to $t_1$.
Thus,
$p'$ sets $\textit{w-sqno}_{p'}$ to a value larger or equal to $t_1$,
before sending a read response
$\langle \text{Ack-R}, \textit{r-sqno}, \langle \textit{w-sqno}_{p'}, \textit{val} \rangle\rangle$
to $p_j$.
Since $p_j$ picks the value associated with the maximum sequence number $t_2$,
it follows that $t_2 \geq \textit{w-sqno}_{p'} \geq t_1$.
\end{proof}

We show that the algorithm implements an atomic register,
by explicitly ordering all completed reads and
all invoked writes (even if they are incomplete).
Note that values written by the writer $w$
have distinct write sequence numbers, and
are different from the initial value of the register, denoted $v_0$;
the value of the $k^{th}$ write operation is denoted $v_k$,
$k \geq 1$.

Writes are ordered by the order they are invoked by process $w$;
if the last write is incomplete, we place this write at the end.
Since only one process invokes write, this ordering is well-defined
and furthermore, the values written appear in the order
$v_1, v_2, \ldots$.

Next, we consider reads in the order they complete;
note that this means that non-overlapping operations are considered
in their order in the execution.
A read that returns the value $v_{k-1}$, $k \geq 0$,
is placed before the $k$-th write in the ordering, if this write exists,
and at the end of the ordering, otherwise.
For $k=0$, this means that the read is placed before the first write,
which may be at the end of the order, if there is no write.

%
%

Lemma~\ref{lem:msgexw} implies that this order respects the real-time
order of non-overlapping operations.
Together with Lemma~\ref{lem:rwcrash}, this implies:

\begin{theorem} \label{theorem:non partitionable SWMR}
If a system is not $f$-partitionable then \algoref{swmr}
implements an atomic SWMR register, in the presence of $f$ failures.
\end{theorem}

\subsection{Impossibility of Implementing a Register in a Partitionable System}

The impossibility proof holds even if only \emph{regular}
register~\cite{Lamport86} is implemented.
In a regular register,
a read should return the value of a {\sc Write} operation
that either overlaps it, or immediately precedes it.
The proof is similar to the one in~\cite{Hadzilacos+:M&M:OPODIS:2019},
where they show that a SWMR register cannot be implemented in the
m\&m model if more than $\tL$ processes may fail,
but it is adapted to handle access restrictions.

To carry out the impossibility proof (and later one),
we need additional definitions.
A \emph{configuration} $C$ is a tuple with a state for each process,
a value for each shared register,
and a set of messages in transit (sent but not received)
between any pair of processes.
A \emph{schedule} is a sequence of process identifiers.
For a set of processes $P$, a schedule is \emph{$P$-free} if no
process from $P$ appears in the schedule;
a schedule is \emph{$P$-only} if only processes from $P$ appear
in the schedule.
An \emph{execution} $\alpha$ is an alternating sequence of configurations
and \emph{events}, where each event is a step by a single process that
takes the system from the preceding configuration to the following
configuration.
In a step, a process either accesses the shared memory
(read or write) or receives and sends messages.
Additionally, a step may involve the invocation of a higher-level
operation.
A schedule is associated with the execution in a natural way;
this induces notions of $P$-free and $P$-only executions.

\begin{theorem}\label{theorem:f-partitionable impossibility}
If a system is $f$-partitionable then for at least $n-f$ processes $p$,
there is no implementation of a regular SWMR register
with writer $p$ in the presence of $f$ failures.
\end{theorem}

\begin{proof}
Assume, by way of contradiction, that there is an implementation of a
regular register $R$ that tolerates $f$ failures.
Since the system is $f$-partitionable,
there are two disjoint sets of processes $P$ and $P'$,
each of size $n-f$, such that $\canreadcomm[not]{P'}{P}$.
This means that there are no two processes $p\in{P}$
and $p'\in{P'}$, such that
$p'$ can read from a memory and $p$ can write to that same memory.
Let $Q$ be the set of nodes not in $P \cup P'$.
Since $|P|, |P'| = n-f$, it follows that $| P \cup Q | = | P' \cup Q | = f$.
Therefore, the implementation terminates in every
execution in which all processes in $P \cup Q$ or in $P' \cup Q$ crash,
i.e., in $(P \cup Q)$-free executions
and in $(P' \cup Q)$-free executions.

Assume that the writer $w$ is in $P$.

To prove the theorem, we execute a write operation by $w$,
followed by a read operation by some process in $P'$,
while partitioning them.
Namely, we delay all messages sent from $P$ until the read completes.
The processes in $P'$ cannot access any of the shared memories
that processes in $P$ write to,
implying that the read operation is unaware of the write operation,
violating the register specification.
Formally, this is done by constructing three executions.

The first execution, $\alpha_1$, is a $(P' \cup Q)$-free execution.
At time $0$, the writer $w$ invokes a {\sc Write}($v$) operation,
$v \neq v_0$ (the initial value of the implemented register).	
Since at most $f$ processes fail in the execution $\alpha_1$,
and the implementation tolerates $f$ failures,
the write must complete at some time $t_w^e$.

The second execution, $\alpha_2$, is a $(P \cup Q)$-free execution.
At some time $t_r^s > t_w^e$, a process $q \in P'$ invokes
a {\sc Read} operation.
Since at most $f$ processes fail in the execution $\alpha_2$,
and the implementation tolerates $f$ failures,
the read completes at some time $t_r^e$.
Since there is no write in $\alpha_2$, the read returns $v_0$.
	
Finally, the third execution $\alpha_3$
merges $\alpha_1$ and $\alpha_2$.
The execution is identical to $\alpha_1$ from time 0 to $t_w^e$
and to $\alpha_2$ from time $t_r^s$ to $t_r^e$.
All messages sent from processes in $P$ to processes in $P'$
and from processes in $P'$ to processes in $P$
are delivered after time $t_r^e$.
Processes in $Q$ take no steps in $\alpha_3$.
Since 
$q$ does not receive any messages from a process in $P$,
and since processes in $P'$ cannot read what
processes in $P$ write to the shared memory,
$q$ returns $v_0$ as in $\alpha_2$,
violating the specification of a regular register. 
Since the writer can be any process in the set $P$, there are 
at least $n-f$ processes for which the theorem holds.
\end{proof}

Note that even when the system is $f$-partitionable there could
be writers for which a SWMR register can be implemented.
A simple example is when we have single memory $\mu$,
writable by some process $w$ and readable by all processes.
A SWMR register writable by $w$ and readable by all the other
processes can be trivially implemented,
regardless of whether the system is partitionable.
However, a SWMR register writable by other processes cannot be implemented.



If the shared memory is without access restrictions,
the $\rightarrow$ relation is symmetric,
which implies the next observation:

\begin{observation}
If a system is $f$-partitionable and the shared memory has no
access restrictions then there is no implementation of a regular
SWMR register in the presence of $f$ failures, for any writer $p$.
\end{observation}

Since all the processes can write to a MWMR register,
we can show the following theorem ,
with the same proof as \theoremref{f-partitionable impossibility}.

\begin{theorem}
If a system is $f$-partitionable then there is no implementation of
a regular MWMR register in the presence of $f$ failures.
\end{theorem}


%% file: tasks.tex
\section{Solving Other Problems in Mixed Systems}
\label{sec:tasks}

\subsection{Constructing Other Read/Write Registers}

The atomic SWMR register presented in the previous section
can be used as a basic building block for implementing
other shared-memory objects.
Recall that if a system is not $f$-partitionable
(i.e., $f\leq \fnonpartionable$),
a SWMR register can be implemented so that each operation
takes $O(1)$ time, $O(n)$ messages,
$O(\regsno)$ SWMR shared-memory registers
and $O(\readno)$ SWMR shared-memory accesses.
Given a shared-memory algorithm that uses $O(r)$ SWMR registers and
has $O(s)$ step complexity,
it can be simulated with $O(s)$ round-trips, $O(ns)$ messages,
and $O(\readno s)$ shared-memory accesses.
The simulation requires $O(\regsno r)$ SWMR shared-memory registers.
(Recall that $\regsno = \sum_{\text{process $p$}}|W_p|$ and
$\readno = \sum_{\text{process $p$}}{\sum_{\mu\in R_p}{|W_\mu|}}$.)

An atomic \emph{multi-writer multi-reader} (MWMR) register can be
built from atomic SWMR registers~\cite{VitanyiA86};
each read or write requires $O(n)$ round-trips,
$O(n^2)$ messages, $O(\regsno n)$ SWMR shared registers and $O(\alpha n)$ shared-memory accesses.

Atomic snapshots can also be implemented using
SWMR registers~\cite{AttiyaR98};
each scan or update takes $O(n\log n)$ round-trips,
$O(n^2\log n)$ messages, $O(\regsno n)$ SWMR shared registers
and $O(\readno n \log n)$ shared-memory accesses.


\subsection{Batching}
\label{sec:batching}

A simple optimization is \emph{batching} of read requests,
namely reading the registers of several processes simultaneously.
Batching is useful when each process replicates a register
for each other process---not for just one writer.
A process $p$ writes the value $v$ using {\sc{Write}$_p$}($v$).
In {\sc{Collect}}(), a process $p$ sends read requests
for all these registers together
instead of sending $n$ separate read requests
(for the registers of all processes), one after the other.
When a process $q$ receives the batched request from $p$,
it replies with a vector containing the values of all registers in
a single message, rather than sending them separately.
Process $p$ waits for vectors from $n-f$ processes,
and picks from them the latest value for each other process.
Finally, the reader does a write-back of this vector.
An invocation of {\sc{Collect}}() returns a vector $V$
with $n$ components, one for each process.
Each component contains a pair of a value with a sequence number.
For every process $p_i$,
$V[i]$ is the entry in the vector corresponding to $p_i$'s value.
A vector $V_1$ \emph{precedes} a vector $V_2$ if the sequence number
of each component of $V_1$ is smaller than or equal to the corresponding
component of $V_2$.
Pseudocode for the algorithm appears in Algorithm~\ref{alg:batching}.

\begin{algorithm}
	\caption{Batching, code for process $p$.}	
	\label{alg:batching}
	\begin{algorithmic}[1]
			\\ []{\bf Local Variables:}\\[]
			 {\textit{w-sqno}}: int, initially 0 \Comment{write sequence number}\\[]
			 {\textit{r-sqno}}: int, initially 0 \Comment{read sequence number}\\[]
			 \textit{val}: vector of size $n$, with all entries initially $0$ \\[]
			 \textbf{for each} $i\in[n]$: \\[]
			 \quad {\textit{last-sqno[$i$]}}: int, initially 0 \Comment{last write sequence number observed}
			  \\ []{\bf Shared Variables:}  \\[]
			  \textbf{for each} process $p$, memory $\mu\in{W_p}$ and $i\in[n]$: \\[]
			  \quad {\textit{R$_\mu$[p][$i$]}}: $\langle$int, int$\rangle$, initially $\langle0, v_0\rangle$
                    \Comment{writable by $p$ and readable by all processes that can read
                    from $\mu$, i.e., all the processes in $R_\mu$}\\[]
		 \hrulefill
			\makeatletter\setcounter{ALG@line}{0}\makeatother
			\item[] \blank{-.7cm} \textbf{write$_i$(\textit{sqno, v}):}
                \Comment{helper function that writes the value of process $p_i$
                            to the shared memory}
			\If {\textit{(sqno $>$ last-sqno[$i$])}}
			\State \textit{last-sqno[$i$]} = \textit{sqno}
			\ForEach {$\mu\in{W_p}$}
			    \Comment{write value and sequence number to every register $p$ can write to}
			\State \textit{R$_\mu$[p][$i$]} = \textit{$\langle$sqno, v}$\rangle$
			\EndFor
			\EndIf
			\item[] \vspace{-0.2cm}
			\item[] \blank{-.7cm} {\sc{Write}}$_p$(\textit{v}):
			\State \textit{w-sqno} = \textit{w-sqno} + 1 \Comment{increment the write sequence number}
			\State \textit{acks} =  {msg\_exchange}$\langle$W\textit{, w-sqno, v}$\rangle$
			\State return
			\item[] \vspace{-0.2cm}
			\State \textbf{Upon receipt of a $\langle$W\textit{, sqno, v}$\rangle$ message from process \textit{p$_i$}:}
			\State write$_i$(\textit{sqno, v})
			\State send $\langle$Ack-W\textit{, sqno}$\rangle$ to process \textit{p$_i$}
			\item[] \vspace{-0.2cm}
			\item[] \blank{-.7cm} {\sc{Collect}}():
			\State \textit{r-sqno} = \textit{r-sqno} + 1 \Comment{increment the read sequence number}
			\State \textit{set\_of\_tuples} =
                {msg\_exchange}$\langle$R\textit{, r-sqno, $\bot$}$\rangle$
            \ForEach{$i\in [n]$}
            \State \textit{set\_of\_tuples$_i$} = \textit{set\_of\_tuples[$i$]} \Comment{all the tuples received for process $p_i$}
            \State \textit{$\langle$seq$_i$, val$_i$}$\rangle$ = max\textit{(set\_of\_tuples$_i$)} \Comment{maximum \textit{$\langle$seq, val$\rangle$} recieved for process $p_i$}
            \State \textit{val[$i$]} = $\langle$\textit{seq$_i$, val$_i$}$\rangle$
            \EndFor
			\State \textit{acks} = {msg\_exchange}$\langle$WB\textit{, r-sqno, val}$\rangle$\Comment{write back}
			\State return \textit{val}
			\item[] \vspace{-0.2cm}
			\State \textbf{Upon receipt of a $\langle$R\textit{, r-sqno, -}$\rangle$ message from process \textit{q}:}
            \ForEach{$i\in [n]$}
            \State \textit{$\langle$w-seq$_i$, w-val$_i$}$\rangle$ =
                max\{\textit{$\langle$w-seq, w-val}$\rangle$ : $\mu\in{R_p\cap{W_q}}$ and \textit{R$_\mu$[q][$i$] = $\langle$w-seq, w-val}$\rangle$\} 
            \State \textit{val[$i$]} = \textit{$\langle$w-seq$_i$, w-val$_i$}$\rangle$
            \EndFor
			\State send $\langle$Ack-R\textit{, r-sqno, val}$\rangle$ to process \textit{q}
			\item[] \vspace{-0.2cm}
			\State \textbf{Upon receipt of a $\langle$WB\textit{, sqno, val}$\rangle$ message from process \textit{q}:}
			\ForEach{$i\in [n]$}
			\State $\langle$\textit{seq, v}$\rangle$ = \textit{val[$i$]}
            \State write$_i$(\textit{seq, v})
            \EndFor
			\State send $\langle$Ack-WB\textit{, sqno}$\rangle$ to process \textit{q}
		\end{algorithmic}
\end{algorithm}

Batching provides a \emph{regular collect}~\cite{AttiyaFG2002},
as stated in the next lemmas, which we bring without proof,
as they are similar to the SWMR register correctness proofs
(Section \ref{swmr-correct}).

\begin{lemma}
Let $t_2$ be the largest sequence number returned in a
read msg\_exchange for the value of process $p_i$ in a {\sc{Collect}}
operation by reader $r$,
and assume that the msg\_exchange starts after the completion of
a write msg\_exchange in a {\sc{Write}}$_{p_i}$ operation by process $p_i$
with sequence number $t_1$, then, $t_1$ $\leq$ $t_2$.
\end{lemma}

\begin{lemma}
For every $1\leq i \leq n$, let $t_i$ be the largest sequence number returned
in a read msg\_exchange for the value of process $p_i$ in a {\sc{Collect}}
operation by reader $r$, and assume that the msg\_exchange starts
after the completion of a write back msg\_exchange writing the value \textit{val},
then \textit{seq$_i$} $\leq$ $t_i$, where $\langle$\textit{seq$_i$, -}$\rangle$ = \textit{val[$i$]}.
\end{lemma}

\begin{lemma}
Let \textit{V$_1$} be a vector returned by a {\sc{Collect}} operation $c_1$,
and let \textit{V$_2$} be a vector returned by a {\sc{Collect}} operation $c_2$,
such that $c_1$ returns before the invocation of $c_2$,
then $V_1$ precedes $V_2$.
\end{lemma}

Batching reduces the number of round-trips and messages,
and shared-memory registers and accesses,
but increases the size of messages and registers.
With batching,
an operation on a MWMR register requires $O(1)$ round-trips,
$O(n)$ messages, $O(\regsno)$ SWMR shared registers and $O(\alpha)$ shared-memory accesses,
when each process saves all the writers values
in a single SWMR register.
Batching can also be applied to atomic snapshots,
so that each scan or update takes $O(\log n)$ round-trips,
$O(n\log n)$ messages, $O(\regsno)$ SWMR shared registers and
$O(\readno \log n)$ shared-memory accesses.
%
%
These results are summarized in Table~\ref{tab:sim-ds}.

\begin{table}
    \centering
    \begin{tabular}{ |c|c|c|c|c| }
\hline
 & time & no.~messages & no.~registers & shared-memory accesses \\
\hline\hline
SWMR atomic register & $O(1)$ & $O(n)$ & $O(\regsno)$ & $O(\readno)$ \\
\hline
MWMR atomic register~\cite{VitanyiA86} & $O(1)$ & $O(n)$ & $O(\regsno)$ & $O(\readno)$ \\
\hline
atomic snapshot~\cite{AttiyaR98} & $O(\log n)$ & $O(n\log n)$ & $O(\regsno)$ & $O(\readno \log n)$ \\
\hline
\end{tabular}
\caption{Complexity bounds for simulated data structures (with batching)
in non-partitionable systems}
\label{tab:sim-ds}
\end{table}

Regular collects can be used in the following building block,
where a process repeatedly call collect, and returns a vector of
values if it has received it \emph{twice} (in two consecutive collects).
Two vectors are the same if they contain the same sequence numbers
in each component.
Process $p$ can write the value $v$ using procedure {\sc{Write}$_p$}($v$),
and repeatedly double collect all the processes current values using
the procedure \bb().
An invocation of \bb() returns a vector $V$
with $n$ components, one for each process.
Each component contains a pair of a value with a sequence number.
Pseudocode for this procedure appears in
Algorithm~\ref{alg:building-block}.

\begin{algorithm}[tb]
	\caption{Building block implementation, code for process $p$.}	
	\label{alg:building-block}
	\begin{algorithmic}[1]
			\item[] \blank{-.7cm} \bb():
			\State \textit{V$_1$} = {\sc{Collect}}()
			\While{\textit{true}}
			\State \textit{V$_2$} = \textit{V$_1$}; \textit{V$_1$} = {\sc{Collect}}()
			\If{\textit{V$_1$} = \textit{V$_2$}}
			\State return \textit{V$_1$}
			\EndIf
			\EndWhile
		\end{algorithmic}
\end{algorithm}

The next lemma is easy to prove if writes are atomic, but it holds
even without this assumption.

\begin{lemma}\label{lem:global-order}
If $V_1$ and $V_2$ are vectors returned by two pairs of successful
double collects by processes $p_{i_1}$ and $p_{i_2}$, respectively,
then either $V_1$ precedes $V_2$ or $V_2$ precedes $V_1$.
\end{lemma}

\begin{proof}
Note that either
the first collect of $p_{i_1}$ completes
before the second collect of $p_{i_2}$ starts or
the first collect of $p_{i_2}$ completes
before the second collect of $p_{i_1}$ starts.
We consider the first case, as the second one is symmetric.
Consider the sequence numbers $s_1$ and $s_2$ in $V_1[j]$ and $V_2[j]$,
for any process $p_j$.
Since the first collect of $p_{i_1}$ completes before
the second collect of $p_{i_2}$ starts,
the regularity of collect ensures that $s_1 \leq s_2$.
\end{proof}


This building block may not terminate (even if the system is not
$f$-partitionable), due to continuous writes.
However, if two consecutive collects are not equal then some sequence
number was incremented, i.e., a write by some process is in progress.

\subsection{Approximate Agreement}

In the \emph{approximate agreement} problem with parameter $\epsilon > 0$,
all processes start with a real-valued input and must decide on an output value,
so any two decision values are in distance at most $\epsilon$
from each other (\emph{agreement}),
and any decision value is in the range of all initial values (\emph{validity}).

There is a wait-free algorithm for the approximate agreement problem
in the shared-memory model, which uses only SWMR registers \cite{AttiyaLS94}.
This algorithm can be simulated if the system is not $f$-partitionable,
and at most $f$ processes fail.
Similarly to randomized consensus,
it can be shown that this problem is unsolvable in partitionable systems.

\begin{theorem}
\label{theorem:approx-agreement-impossibilty}
If a system is $f$-partitionable then approximate agreement is
unsolvable in the presence of $f$ failures.
\end{theorem}

\begin{proof}
Assume, by way of contradiction,
that there is an approximate agreement algorithm.
Since the system is $f$-partitionable,
there are two disjoint sets of processes $P$ and $P'$,
each of size $n-f$, such that $\bothcanreadcomm[not]{P}{P'}$.
Assume, without loss of generality, that $\canreadcomm[not]{P'}{P}$.
This means that any process in ${P'}$ cannot read memories
that processes in $P$ can write to.
Let $Q$ be the processes not in $P \cup P'$.
Since $|P|, |P'| = n-f$,
it follows that $| P \cup Q | = | P' \cup Q | = f$.

To prove the theorem, we construct three executions.

The first execution, $\alpha_1$, is a $(P' \cup Q)$-free execution,
in which all processes in $P$ have initial value $0$.
Since at most $f$ processes fail in the execution and
the implementation can tolerate $f$ failures,
all processes in $P$ must eventually terminate by some time $t_1$.
By validity, every process in $P$ decides $0$.

The second execution, $\alpha_2$, is a $(P \cup Q)$-free execution,
in which all processes in $P'$ have initial value $2\epsilon$.
Since at most $f$ processes fail in the execution and
the implementation can tolerate $f$ failures,
all processes in $P'$ must eventually terminate by some time $t_2$.
By validity, every process in $P$ decides $2\epsilon$.

The third and final execution, $\alpha_3$,
combines $\alpha_1$ and $\alpha_2$.
The initial value of processes in $P$ is $0$,
and the initial value of processes in $P'$ is $2\epsilon$.
Processes in $Q$ have arbitrary initial values,
and they take no steps in $\alpha_3$.
The execution is identical to $\alpha_1$ from time $0$ until time $t_1$,
and to $\alpha_2$ from this time until time $t_1 + t_2$.
All messages sent between processes in $P$ and processes in $P'$
are delivered after time $t_1 + t_2$.

Since processes in $P'$ do not take steps in $\alpha_3$ until time $t_1$,
all processes in $P$ decides $0$ until this time,
as in execution $\alpha_1$.
Processes in $P'$ cannot communicate with processes in $P$,
either by messages or read what they wrote to the shared memory,
therefore all processes in $P'$ decides $2\epsilon$,
as in execution $\alpha_2$.
Since $|0-2\epsilon|>\epsilon$, this violates the agreement property.
\end{proof}

\subsection{Renaming}
In the \emph{$M$-renaming} problem, processes start with unique
\emph{original} names from a large namespace $\{1,...,N\}$,
and the processes pick distinct \emph{new} names from
a smaller namespace $\{1,...,M\}$ ($M < N$).
To avoid a trivial solution, in which a process $p_i$ picks its
index $i$ as the new name, we require \emph{anonymity}:
a process $p_i$ with original name $m$ performs the same as
process $p_j$ with original name $m$.

Employing the SWMR register simulation in a $(2n-1)$-renaming
algorithm~\cite{AttiyaF01} yields
an algorithm that requires $O(n \log n)$ round-trips,
$O(n^2 \log n)$ messages,
$O(\regsno n^4)$ shared registers
and $O(\readno n \log n)$ shared-memory accesses.
The number of registers can be reduced to $O(\regsno)$,
at the cost of increasing their size.

This algorithm assumes that the system is not $f$-partitionable
and at most $f$ processes fail.
The next theorem shows that this is a necessary condition.

\begin{theorem} \label{theorem:renaming-impossibility}
If a system is $f$-partitionable then renaming is unsolvable
in the presence of $f$ failures.
\end{theorem}

\begin{proof}
Assume, by way of contradiction, that there is a renaming algorithm.
Since the system is $f$-partitionable,
there are two disjoint sets of processes $P$ and $P'$, each of size $n-f$,
such that $\canreadcomm[not]{P'}{P}$.
Denote $P=\{p_{i_1},...,p_{i_{n-f}}\}$ and $P'=\{p'_{i_1},...,p'_{i_{n-f}}\}$.
Let $Q$ be the set of processes not in $P \cup P'$.
Since $|P|, |P'| = n-f$,
we have that $| P \cup Q | = | P' \cup Q | = f$.

Given a vector $I$ of $n-f$ original names, denote by $\alpha(I,P)$ the
$P$-only execution in which processes in $P$ have original names $I$:
processes in ($P'\cup Q$) crash and take no step,
and processes in $P$ are scheduled in round-robin.
Since at most $f$ processes fail in $\alpha(I,P)$,
eventually all processes in $P$ pick distinct new names,
say by time $t(I)$.
Note that by anonymity, the same names are picked in the execution
$\alpha(I,P')$, in which $p'_{i_j}$ starts with the same original name
as $p_{i_j}$ and takes analogous steps.

Consider $\alpha(I_i,P)$, for any possible set of original names.
The original name space can be picked to be big enough to ensure
that for two \emph{disjoint} name assignments, $I_1$ and $I_2$,
some process $p_{i_j}\in P$ decides the same new name $r$ in the executions
$\alpha(I_1,P)$ and $\alpha(I_2,P)$.

Denote $\alpha_1 = \alpha(I_1,P)$ and $\alpha_2=\alpha(I_2,P')$,
namely, the execution in which processes in $P'$ replace the
corresponding processes from $P$.
The anonymity assumption ensures that $p_{i_j}'$ decides on $r$,
just as $p_{i_j}$ decides on $r$ in $\alpha(I_1,P)$ and $\alpha(I_2,P)$.

The execution $\alpha_3$ combines $\alpha_1$ and $\alpha_2$, as follows.
Processes in $Q$ take no steps in $\alpha_3$.
The original names of processes in $P$ are $I_1$,
and original names of processes in $P'$ are $I_2$.
The execution is identical to $\alpha_1$ from time 0 until time $t(I_1)$,
and to $\alpha_2$ from this time until time $t(I_1)+t(I_2)$.
All messages sent from processes in $P$ to processes in $P'$
and from processes in $P'$ to processes in $P$ are delivered after
time $t(I_1)+t(I_2)$.

In $\alpha_3$,
processes in $P$ do not receive messages from processes in $P' \cup Q$.
Furthermore, $\canreadcomm[not]{P'}{P}$; i.e., processes in $P'$
cannot read what processes in $P$ wrote to the shared memory.
Hence, $\alpha_3$ is indistinguishable to $p_{i_j}$ from $\alpha_1$,
and hence, it picks new name $r$.
Similarly, $\alpha_3$ is indistinguishable to $p'_{i_j}$ from $\alpha_2$,
and hence, it also picks new name $r$,
which contradicts the uniqueness of new names.
\end{proof}

\subsection{Consensus}
\label{sec:consensus}

In the \emph{consensus} problem, a process starts with an
input value and decides on an output value,
so that all processes decide on the same value (\emph{agreement}),
which is the input value of some process (\emph{validity}).

With a standard \emph{termination} requirement, it is well known
that consensus cannot be solved in an asynchronous system~\cite{FLP}.
This result holds whether processes communicate through shared memory
or by message passing, and even if only a single process fails.
However, consensus can be solved if the termination condition is
weakened, either to be required only with high probability
(\emph{randomized consensus}),
or to hold when it is possible to eventually detect failures
(using a \emph{failure detector}),
or to happen only under fortunate situations.

There are numerous shared-memory randomized consensus algorithms,
which rely on read / write registers,
or objects constructed out of them.
Using these algorithms together with linearizable register implementations
is not obvious since linearizability does not preserve
\emph{hyperproperties}~\cite{GolabHW2011,AttiyaE2019}.
It has been shown~\cite{HadzilacosHT2020arxiv}
that the ABD register implementation~\cite{AttiyaBD1995}
is not \emph{strongly linearizable}~\cite{GolabHW2011}.
This extends to the mixed-model register implementations,
as ABD is a special case of them.

Hadzilacos et al.~\cite{HadzilacosHT2020regular} have proved that
the simple randomized consensus algorithm of~\cite{AspnesH1990}
works correctly with \emph{regular} registers,
and used it to obtain consensus in m\&m systems~\cite{HadzilacosHT2020arxiv}.
Their algorithm inherits exponential complexity from the
simple algorithm of~\cite{AspnesH1990},
which employs independent coin flips by the processes.

A \emph{weak shared coin} ensures that,
for every value $v\in\{-1,1\}$,
all processes obtain the value $v$ with probability $\delta > 0$;
$\delta$ is called the \emph{agreement parameter} of the coin.
Using a weak shared coin with a constant agreement parameter in the
algorithm of~\cite{AspnesH1990},
yields an algorithm with constant number of rounds.
Thus, randomized consensus inherits the complexity of the
weak shared coin.

Here, we explain how to use \bb()
to emulate the weak shared coin of~\cite{AspnesH1990},
following~\cite{BarNoyD1993}.
This holds when there are at most $f$ failures,
if the system is not $f$-partitionable.

\begin{algorithm}[t]
	\caption{Weak shared coin, based on \cite{AspnesH1990}}	
	\label{alg:coin}
	\begin{algorithmic}[1]
			\\ []{\bf Local Variables:}\\[]
			 {\textit{my-counter}}: int, initially 0 \\[]
			 $V$: vector of size $n$, with all entries initially $0$ \\[]
		 \hrulefill
			\makeatletter\setcounter{ALG@line}{0}\makeatother
			\item[] \blank{-.7cm} {\sc{Coin}}() --- \textbf{Code for process $p$:}
			\While{\textit{true}}
			\State \textit{my-counter = my-counter} + flip()
			\State {\sc{Write}}$_p$(\textit{my-counter})
			\State $V$ = \bb()
			\If {$\text{sum}(V) \geq c\cdot n$} {return 1}
			\ElsIf {$\text{sum}(V) \leq -c\cdot n$} {return -1} \label{lin:tail}
			\EndIf
			\EndWhile
		\end{algorithmic}
\end{algorithm}	

In \algoref{coin},
a process flips a coin using a local function flip(),
which returns the value 1 or -1, each with probably 1/2.
Invoking flip() is a single atomic step.
After each flip, a process writes its outcome in an individual
cumulative sum.
Then it calls \bb() to obtain a vector $V$
with the individual cumulative sums of all processes.
(We assume that the initial value in each component is $0$.)
The process then checks the absolute value of the total sum of the
individual cumulative sums, denoted sum($V$).
If it is at least $c\cdot n$ for some constant $c>1$,
then the process returns its sign.

Intuitively, the only way the adversary can create disagreement on
the outcome of the shared coin is by preventing as many processors
as possible to move the counter in the unwanted direction.
We will show that the adversary cannot ``hide'' more than
$n-1$ coin flips.
(This was originally proved when processes use atomic
writes~\cite{AspnesH1990};
here, we show it holds even when writes are not atomic.)
Therefore, after the cumulative sum is big or small enough the
adversary can no longer affect the outcome of the shared coin,
and cannot prevent the processes from terminating.


Fix an execution $\exec$ of the weak shared coin.
Let $H$ and $T$ be the number of 1 and -1 (respectively) flipped by
all processes after some prefix of the execution.
These numbers are well-defined since the local coin flips are atomic.

\begin{lemma}\label{lem:ah-15}
If $H-T < -(c+1)\cdot n$ (respectively, $H-T>(c+1)\cdot n$) after
some prefix $\exec'$ of the execution,
then a process that invokes \bb() in $\exec'$
and completes it, returns $-1$ (respectively, 1) from {\sc{Coin}}().
\end{lemma}

\begin{proof}
We consider the first case; the other case is symmetric.
Assume $H-T < -(c+1)\cdot n$ after an execution prefix $\exec'$.
Consider the processes invoking \bb()
after $\exec'$, in the order their \bb() return.
Let $p_{j_i}$, $i \geq 1$, be the $i$th process in this order,
and let $V_i$ be the vector returned by its \bb().
We prove, by induction on $i$, that $\text{sum}(V_i) \leq -c\cdot n$,
and hence, $p_{j_i}$ returns $-1$.
(See \figref{ah15-fig}.)

In the base case, $i=1$.
Since a process invokes \bb() after every write,
there can be at most $n$ writes by all processes
after $\exec'$ 
and before the return of \bb() by $p_{j_1}$.
These writes can be of flips tossed in $\exec'$,
or flips tossed after $\exec'$.
The value $p_{j_1}$ 
returns in $V_1[k]$ can be different than the local value
of \textit{my-counter} at $p_k$ after $\exec'$
in one of two cases:
\begin{enumerate}
    \item If $p_k$ tossed a coin in $\exec'$,
    and the write of this coin does not complete before
    \bb() by $p_{j_1}$ returned
    and $V_1$ does not contain this value.
    \item If $p$ tossed a coin after $\exec'$,
    and $V_1$ contains this value.
    Hence, the write of this coin has started before
    \bb() by $p_{j_1}$ returns.
\end{enumerate}
Since there is at most one write by process $p_k$ after $\exec'$
until the \bb() by $p_{j_1}$ returns,
the regularity of collect implies that $p_{j_1}$ reads a value different
by at most 1 from the value of $p_k$'s counter after $\exec'$.
Therefore, $p_{j_1}$ reads a value bigger by at most $n$
from the value of $H-T$ after $\exec'$ and $\text{sum}(V_1)< -c\cdot n$,
thus $p_{j_1}$ decides -1 in Line~\ref{lin:tail}.

Inductive step: Assume that for $i > 1$,
processes $p_{j_1},...,p_{j_{i-1}}$ decide after their
\bb() invocation returns.
These processes have at most one write after $\exec'$ and before
they invoke \bb().
By the induction hypothesis, they return from {\sc{Coin}}()
after their \bb() invocation returns.
Therefore, these process have no additional writes after $\exec'$.
For any other process $p_k$,
there is at most one write after $\exec'$ until the
\bb() by $p_{j_i}$ returns,
since there must be a \bb() operation
by $p_k$ between any two writes by this process.
Therefore, there are at most $n$ writes after $\exec'$ until $p_{j_i}$
returns from \bb().
It follows that $p_{j_i}$ reads at most $n$ additional values from $H-T$
and returns $-1$ from {\sc{Coin}}().
\end{proof}

\begin{figure}
		\centering
		\captionsetup{justification=centering}
		\scalebox{.85}{\input{figs/ah_15.pdf_t}}
		\caption{Illustration of \lemref{ah-15}}
		\label{fig:ah15-fig}
\end{figure}

\begin{lemma}\label{lem:ah-16}
If process $p$ returns 1 (respectively, -1) from the shared coin,
then $H-T>(c-1)\cdot n$ (respectively, $H-T<-(c-1)\cdot n$)
at some point during its last call to \bb().
\end{lemma}

\begin{proof}
We consider the first case; the other case is symmetric.
Consider the last pair of collects in the last \bb()
invocation before process $p$ returns,
and assume they return a vector $V$.
Assume $p$ misses a write by some process $p_k$
that overlaps the first collect,
i.e., the sequence number of this write is smaller than
the corresponding sequence number in $V$.
Then $p_k$'s write overlaps $p$'s first collect,
and it returns after the second collect starts.
(Otherwise, the regularity of collect implies that the second
collect returns this write by $p_k$, or a later one,
contradicting the fact it is equal to the first collect.)
Therefore, each process has at most one write that overlaps
the first collect and is missed by the first collect.

$V[k]$ is different from the local value of \textit{my-counter} at $p_k$
when the first collect completes,
if either $p_k$ did not start a write for a tossed coin
or a write by $p_k$ is pending at this point.
In the first case, $p$ reads all $p_k$'s previous writes.
In the second case, $p$ reads all previous writes by $p_k$,
except maybe the pending write.
Thus, from the regularity of collect, the value of $p_k$ read by $p$
differs by at most 1 from its value when
the first collect of $p$ completes.
Therefore, the sum of $V$ differs by at most $n-1$ values from
the value of $H-T$ when the first collect completes.
Since $p$ returns 1, $\text{sum}(V) \geq c\cdot n$, and hence,
$H-T>(c-1)\cdot n$ when the first collect of $p$ completes.
\end{proof}


\begin{lemma}
The agreement parameter of the weak shared coin algorithm is $(c-1)/2c$.
\end{lemma}

\begin{proof}
We show that with probability $(c-1)/2c$ all the processes decide -1;
the proof for 1 is analogous.
\lemref{ah-16} implies that $H-T$ reaches $(c-1)\cdot n$ before
a process decides 1.
By~\lemref{ah-15}, after $H-T$ drops below $-(c+1)\cdot n$ there are
at most $n$ more \bb() invocations in the execution.
Therefore,
there are at most $n$ more writes in the rest of the execution.
After every coin flip there is a write by the process tossing the coin.
As there are at most $n$ additional writes,
there are also at most $n$ more coin flips in the rest of the execution.
Thus, $H-T<-c\cdot n< (c-1)\cdot n$ until the end of the execution.
It follows that no process can decide 1 if
$H-T$ drops below $-(c+1)\cdot n$ before reaching $(c-1)\cdot n$.
In addition, all processes eventually decide -1, since otherwise,
a process that has not decided earlier invokes
a \bb() operation.

In the worst case, the adversary can make some undecided process
decide 1 if $H-T$ reaches $(c-1)\cdot n$.
Viewing the value of $H-T$ as a random walk~\cite{Feller57},
starting at the origin with the absorbing barriers at $-(c+1)\cdot n$
(all decide -1) and at $(c-1)\cdot n$ (some may decide 1).
By a classical result of random walk theory,
the probability of reaching $-(c+1)\cdot n$ before $(c-1)\cdot n$ is
$$\frac{(c-1)\cdot n}{(c+1)\cdot n + (c-1)\cdot n}
= \frac{(c-1)\cdot n}{2cn} = \frac{c-1}{2c} .$$
\end{proof}

The next theorem can be proved along similar lines to~\cite{AspnesH1990}.

\begin{theorem}
For a constant $c>1$, the expected number of coin flips in an execution
of the weak shared coin is $O(n^2)$.
\end{theorem}

Since the expected number of coin flips is $O(n^2)$,
the expected number of write and building block invocations
is also $O(n^2)$.
The total number of collect operations in these building block invocations
for all the processes is $O(n^3)$ in expectation.
This is because a double collect fails only when another coin is written,
and thus each write can cause at most $n-1$ double collects to fail. 
Therefore, there are at most $O(n^3)$ failed double collects 
and $O(n^2)$ successful double collects, 
all part of $O(n^2)$ building block invocations by all the processes.
Therefore, the weak shared coin requires $O(n^3)$ round-trips,
$O(n^4)$ messages and $O(\readno n^3)$ shared-memory accesses in expectation
and uses $O(\regsno)$ registers.
Plugging the weak shared coin in the overall algorithm of~\cite{AspnesH1990},
proved to be correct by~\cite{HadzilacosHT2020regular},
yields a randomized consensus algorithm with the same expected
complexities as the weak shared coin.

\begin{theorem}
Randomized consensus can be solved in the presence of $f$ failures
if the system in not $f$-partitionable.
The algorithm takes $O(n^3)$ round-trips, $O(n^4)$ messages
and $O(\readno n^3)$ shared-memory accesses in expectation
and uses $O(\regsno)$ registers.
\end{theorem}

Next, we prove that randomized consensus cannot be solved in a
partitionable system, by considering the more general problem
of \emph{non-deterministic $f$-terminating} consensus,
an extension of \emph{non-deterministic solo termination}~\cite{FichHS98}.
This variant of consensus has the usual validity and agreement properties,
with the following termination property:
\begin{description}
\item[Non-deterministic $f$-termination:]
For every configuration $C$, process $p$
and set $F$ of at most $f$ processes,
such that $p \notin F$,
there is an $F$-free execution in which process $p$ terminates.
\end{description}


\begin{theorem}\label{theorem:consensus-impossibility}
If a system is $f$-partitionable then
non-deterministic $f$-terminating consensus is unsolvable.
\end{theorem}

\begin{proof}
Assume, by way of contradiction, that there is an non-deterministic
$f$-terminating consensus algorithm.
Since the system is $f$-partitionable,
there are two disjoint sets of processes $P$ and $P'$,
each of size $n-f$, such that $\canreadcomm[not]{P'}{P}$.
Therefore, there are no two processes $p\in{P}$ and $p'\in{P'}$ so that
$p'$ can read from a memory and $p$ can write to that same memory.
Let $Q$ be the processes not in $P \cup P'$.
Since $|P|, |P'| = n-f$,
it follows that $| P \cup Q | = | P' \cup Q | = f$.

To prove the theorem, we construct three executions.
Consider an initial configuration, in which all processes
in $P$ have initial value $0$.
Since $|P' \cup Q| = f$, non-deterministic $f$-termination
implies there is a $(P' \cup Q)$-free execution,
in which some process $p \in P$ terminates, say by time $t_1$.
Call this execution $\alpha_1$,
and note that only processes in $P$ take steps in $\alpha_1$.
By validity, $p$ decides 0.

In a similar manner, we can get a $(P \cup Q)$-free execution,
$\alpha_2$, in which initial values of all the processes in $P'$ are 1,
and by non-deterministic $f$-termination, some process $p' \in P'$
decides on 1, say by time $t_2$.
Note that only processes in $P'$ take steps in $\alpha_2$.

Finally, the third execution $\alpha_3$ combines $\alpha_1$ and $\alpha_2$.
The initial value of processes in $P$ is 0,
and the initial value of processes in $P'$ is 1.
Processes in $Q$ have arbitrary initial values,
and they take no steps in $\alpha_3$.
The execution is identical to $\alpha_1$ from time 0 until time $t_1$,
and to $\alpha_2$ from this time until time $t_1+t_2$.
All messages sent between processes in $P$ and processes in $P'$
are delivered after time $t_1 + t_2$.
Since processes in $P'$ do not take steps in $\alpha_3$ until time $t_1$,
all processes in $P$ decides $0$, as in $\alpha_1$.
Processes in $P'$ cannot receive messages from processes in $P$
or read what processes in $P$ write to the shared memory,
therefore all processes in $P'$ decides 1, as in execution $\alpha_2$,
violating the agreement property.
\end{proof}


%% file: figs/ah_15.pdf_t
\tikzset{every picture/.style={line width=0.75pt}} 

\begin{tikzpicture}[x=0.75pt,y=0.75pt,yscale=-1,xscale=1]

\draw    (218.5,88.63) -- (218.5,99.63) ;
\draw    (332.5,88.63) -- (332.5,99.63) ;
\draw    (519.5,88.63) -- (519.5,99.63) ;
\draw    (338.5,116.63) -- (657.5,117.13) ;
\draw [shift={(659.5,117.13)}, rotate = 180.09] [color={rgb, 255:red, 0; green, 0; blue, 0 }  ][line width=0.75]    (10.93,-3.29) .. controls (6.95,-1.4) and (3.31,-0.3) .. (0,0) .. controls (3.31,0.3) and (6.95,1.4) .. (10.93,3.29)   ;
\draw    (86.5,95.13) -- (657,94.14) ;
\draw [shift={(659,94.13)}, rotate = 539.9] [color={rgb, 255:red, 0; green, 0; blue, 0 }  ][line width=0.75]    (10.93,-3.29) .. controls (6.95,-1.4) and (3.31,-0.3) .. (0,0) .. controls (3.31,0.3) and (6.95,1.4) .. (10.93,3.29)   ;
\draw    (523.5,163.63) -- (658.5,163.14) ;
\draw [shift={(660.5,163.13)}, rotate = 539.79] [color={rgb, 255:red, 0; green, 0; blue, 0 }  ][line width=0.75]    (10.93,-3.29) .. controls (6.95,-1.4) and (3.31,-0.3) .. (0,0) .. controls (3.31,0.3) and (6.95,1.4) .. (10.93,3.29)   ;
\draw    (265.19,116.63) -- (334.5,116.63) ;
\draw [shift={(336.5,116.63)}, rotate = 540] [color={rgb, 255:red, 0; green, 0; blue, 0 }  ][line width=0.75]    (10.93,-3.29) .. controls (6.95,-1.4) and (3.31,-0.3) .. (0,0) .. controls (3.31,0.3) and (6.95,1.4) .. (10.93,3.29)   ;
\draw    (265.19,116.63) -- (221.5,116.63) ;
\draw [shift={(219.5,116.63)}, rotate = 360] [color={rgb, 255:red, 0; green, 0; blue, 0 }  ][line width=0.75]    (10.93,-3.29) .. controls (6.95,-1.4) and (3.31,-0.3) .. (0,0) .. controls (3.31,0.3) and (6.95,1.4) .. (10.93,3.29)   ;

\draw    (336.82,163.43) -- (519.5,164.22) ;
\draw [shift={(521.5,164.23)}, rotate = 180.25] [color={rgb, 255:red, 0; green, 0; blue, 0 }  ][line width=0.75]    (10.93,-3.29) .. controls (6.95,-1.4) and (3.31,-0.3) .. (0,0) .. controls (3.31,0.3) and (6.95,1.4) .. (10.93,3.29)   ;
\draw    (336.82,163.43) -- (220.5,162.65) ;
\draw [shift={(218.5,162.63)}, rotate = 360.39] [color={rgb, 255:red, 0; green, 0; blue, 0 }  ][line width=0.75]    (10.93,-3.29) .. controls (6.95,-1.4) and (3.31,-0.3) .. (0,0) .. controls (3.31,0.3) and (6.95,1.4) .. (10.93,3.29)   ;

\draw    (391.93,203.63) -- (657.5,203.63) ;
\draw [shift={(659.5,203.63)}, rotate = 540] [color={rgb, 255:red, 0; green, 0; blue, 0 }  ][line width=0.75]    (10.93,-3.29) .. controls (6.95,-1.4) and (3.31,-0.3) .. (0,0) .. controls (3.31,0.3) and (6.95,1.4) .. (10.93,3.29)   ;
\draw    (391.93,203.63) -- (222.5,203.63) ;
\draw [shift={(220.5,203.63)}, rotate = 360] [color={rgb, 255:red, 0; green, 0; blue, 0 }  ][line width=0.75]    (10.93,-3.29) .. controls (6.95,-1.4) and (3.31,-0.3) .. (0,0) .. controls (3.31,0.3) and (6.95,1.4) .. (10.93,3.29)   ;

\draw (141,101.4) node [anchor=north west][inner sep=0.75pt]  [font=\footnotesize]  {$\beta '$};
\draw (154,67.4) node [anchor=north west][inner sep=0.75pt]  [font=\footnotesize]  {${\displaystyle \boldsymbol{H-T >( c-1) \cdot n}}$};
\draw (282,67.4) node [anchor=north west][inner sep=0.75pt]  [font=\footnotesize]  {$p_{j_{1}} \ \text{incovation returns}$};
\draw (470,67.4) node [anchor=north west][inner sep=0.75pt]  [font=\footnotesize]  {$p_{j_{i}} \ \text{incovation returns}$};
\draw (417,73.4) node [anchor=north west][inner sep=0.75pt]    {$...$};
\draw (218,126) node [anchor=north west][inner sep=0.75pt]  [font=\footnotesize] [align=left] {\begin{minipage}[lt]{86.64152000000001pt}\setlength\topsep{0pt}
each process can have
\begin{center}
at most one write
\end{center}

\end{minipage}};
\draw (427,127.4) node [anchor=north west][inner sep=0.75pt]  [font=\footnotesize]  {$p_{j_{1}} \ \text{decides and stops writing}$};
\draw (514,175.4) node [anchor=north west][inner sep=0.75pt]  [font=\footnotesize]  {$p_{j_{i}} \ \text{decides and stops writing}$};
\draw (293,175.4) node [anchor=north west][inner sep=0.75pt]  [font=\footnotesize]  {$p_{j_{i}} \ \text{can have at most one write}$};
\draw (345,213) node [anchor=north west][inner sep=0.75pt]  [font=\footnotesize] [align=left] {each process can have at most one write};

\end{tikzpicture}

%% file: mm.tex
\section{The M\&M Model}\label{section:mm}

In the m\&m model~\cite{Aguilera:PODC:2018,Hadzilacos+:M&M:OPODIS:2019},
the shared memory connections are defined by a
\textit{shared-memory domain} $L$,
which is a collection of sets of processes.
For each set $S\in L$,
all the processes in the set may share any number of registers among them.
Our model when the shared memory has no access restrictions is
a dual of the general m\&m model, and they both capture the same systems.
We say that $L$ is \emph{uniform} if it is induced by an undirected
\emph{shared-memory graph} $G=(V,E)$,
where each vertex in $V$ represents a process $p$.
For every process $p$, $S_p=\{p\}\cup\{q:(p,q)\in{E}\}$,
then $L=\{S_p: \text{$p$ is a process}\}$.
In the uniform m\&m model each memory is associated with a process $p$,
and all the processes in $S_p$ may access it.
That is, a process can access its own memory and the memories of its
neighbors.

In the m\&m model, there are no access restrictions on the shared memory.
Hence, for every process $p$, $|R_p|=|W_p|=|S_p|$.
Therefore, $\regsno = \sum_{\text{process $p$}}{|S_p|} =
\sum_{\text{process $p$}}{d(p)+1}=2|E|+n=O(n^2)$ and $\readno = O(n^3)$,
where $d(p)$ is the degree of process $p$ in the graph.
Substituting into the algorithms presented in Section \ref{sec:tasks},
we obtain polynomial complexity for all of them,
including a polynomial randomized consensus algorithm.
In the general m\&m model,
$\regsno$ and $\readno$ are unbounded.


\begin{definition}[\cite{Hadzilacos+:M&M:OPODIS:2019}]
Given a shared-memory domain $L$ , $\tL$ is the largest integer $f$
such that for all process subsets $P$ and $P'$ of size $n-f$ each,
either $P\cap{P'}\neq{\emptyset}$ or there is a set $S\in L$
that contains both a process from $P$ and a process from $P'$.
\end{definition}

Hadzilacos, Hu and Toueg~\cite{Hadzilacos+:M&M:OPODIS:2019} show that
an SWMR register can be implemented in the m\&m model if and only if
at most $\tL$ process may fail.
Therefore in the m\&m model, $\fnonpartionable = \tL$.
We can see the connection between the two definitions by observing that,
in this model, $\bothcanreadcomm{p}{q}$ if $p=q$ or
there is a set $S\in L$ such that $p,q \in S$.
We simply write $\bothcanreadcomm{}{}$,
since the shared memory has no access restrictions.

The \emph{square} of a graph $G = (V,E)$ is the graph $G^2=(V,E^2)$,
where  
\[ 
E^2= E \cup\{(u,v):
    \text{$\exists{w\in E}$ such that $(u,w)\in E$ and $(w,v)\in E$} .
\]
I.e., there is an edge in $G^2$ between every two vertices that are
in distance at most 2 in the graph $G$.

%
%

\begin{definition}[\cite{Hadzilacos+:M&M:OPODIS:2019}]
\label{def:fg}
Given an undirected graph $G=(V,E)$, $\fG$ is the largest integer $f$
such that for all subsets $P$ and $P'$ of $V$ of size $n-f$ each,
either $P\cap{P'}\neq{\emptyset}$ or $G^2$ has an edge $(u,v)$
such that $u\in{P}$ and $v\in{P'}$.
\end{definition}

In the uniform m\&m model,
$\tL=\fG=\fnonpartionable$~\cite{Hadzilacos+:M&M:OPODIS:2019},
and $\sharedmemcomm{p}{q}$ if $p=q$ or $(p,q)$ is an edge in $G^2$.

\begin{observation}
\label{obs:fnonpartitionisfg}
In the uniform m\&m model, $\fnonpartionable=\fG$.
\end{observation}

\lemref{fnonpartionleqfmajority} proves that 
$\fnonpartionable\leq\fmajority$.
The next lemma shows a graph with $n=9$, 
where $\fnonpartionable < \fmajority$.
Thus, the converse inequality does not necessarily
hold in the uniform m\&m model.
Hence, it also does not hold in the general m\&m model.

\begin{lemma}
\label{fmajorityneqfnonpartion}
There is a shared-memory graph such that $\fnonpartionable < \fmajority$
in the uniform m\&m model.
\end{lemma}

\begin{proof}
In the uniform m\&m model $\fG=\fnonpartionable$.
Consider the graph $G$ presented in \figref{ce1} with $n=9$.
For every set of processes $P$ of size $n-7=2$,
$|{\communicateset{P}}|\geq{5}>4=\lfloor{n/2}\rfloor$,
therefore, $\fmajority=7$.
We next show that $\fnonpartionable<7$.
Consider the following subsets, $P=\{p_7,p_9\}$ and $Q=\{p_3,p_5\}$,
both of size $n-\fmajority$.
$P\cap{Q}=\emptyset$ and
$(p_7,p_3),(p_7,p_5),(p_9,p_3),(p_9,p_5)\notin{G^{2}}$,
thus, $\fnonpartionable=\fG<7=\fmajority$
and we get that $\fnonpartionable\neq{\fmajority}$ in the m\&m model.
\end{proof}

\begin{figure}
	\begin{minipage}[b]{0.49\textwidth}
		\centering
		\captionsetup{justification=centering}
		\scalebox{.45}{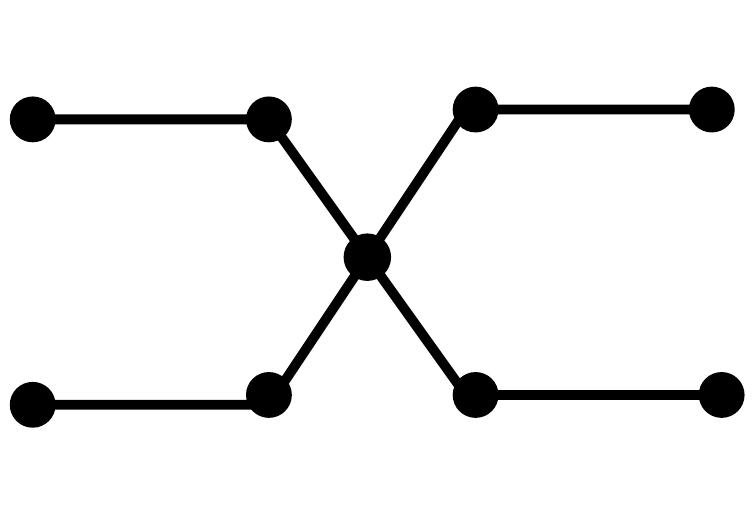}
		\caption{Counter example}
		\label{fig:ce1}
	\end{minipage}
	\hfill
	\begin{minipage}[b]{0.49\textwidth}
		\centering
		\captionsetup{justification=centering}
		\scalebox{.45}{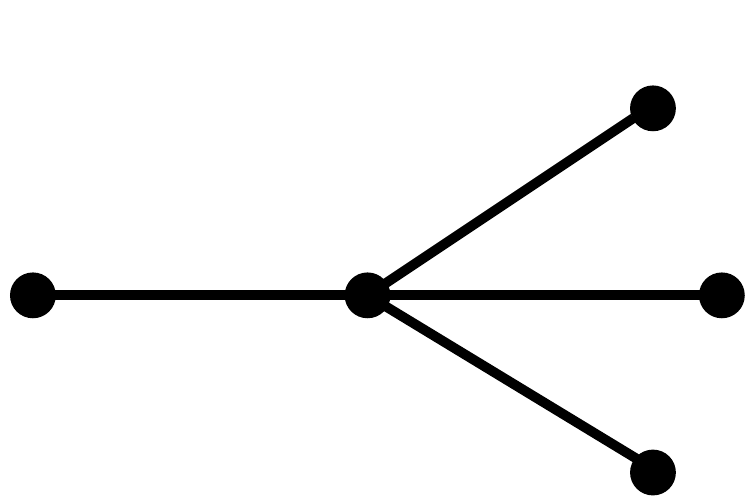}
		\caption{Counter example for $n=5$}
		\label{fig:ce2}
	\end{minipage}
\end{figure}

\begin{definition}[\cite{Aguilera:PODC:2018}]
\label{def:represents}
A process $p$ \emph{represents} itself and all its neighbors,
that is, $\{p\} \cup \{q:(p,q)\in E\}$.
A set of processes $P$ represents the union of all the processes
represented by processes in $P$.
\end{definition}

Aguilera et al.~\cite{Aguilera:PODC:2018} present a randomized
consensus algorithm, called HBO,
which is based on Ben-Or's algorithm~\cite{BenOr1983}.
Like Ben-Or's algorithm, HBO has exponential time and message
complexities.
HBO assumes that the nonfaulty processes represent a majority of the
processes.
Below, we show that the resilience of the HBO algorithm is not optimal.
We first capture the condition required for the correctness of the HBO
algorithm, with the next definition.

\begin{definition}
\label{def:fmandm}
$\fmandm$ is the largest integer $f$ such that every set $P$ of
$n-f$ processes represents a majority of the processes.
\end{definition}

\begin{lemma}
\label{lem:fmandmleqfnonpartion}
$\fmandm\leq{\fnonpartionable}$.
\end{lemma} 

\begin{proof}
Consider two sets of processes $P$ and $Q$, each of size $n-\fmandm$.
Hence, both $P$ and $Q$ represent more than $\lfloor{n/2}\rfloor$ processes,
and hence, they represent a common process $w$.
We consider three cases, and show that in all cases,
$\sharedmemcomm{P}{Q}$, implying that $\fmandm \leq \fnonpartionable$.

If $w\in P \cap Q$ then $P \cap Q \neq{\emptyset}$
and $\sharedmemcomm{P}{Q}$.

If $w\in{P}$ and $w\notin{Q}$
(the case $w\in{Q}$ and $w\notin{P}$ is symmetric).
Since $Q$ represents $w$,
there is a processes $q\in{Q}$ such that $(w,q)\in{G}$,
and hence, $(w,q)$ is also an edge in $G^2$.
Therefore, $\sharedmemcomm{w}{q}$ and $\sharedmemcomm{P}{Q}$.

Finally, if $w \notin P$ and $w \notin Q$,
then since both sets $P$ and $Q$ represent $w$,
there is a process $p\in{P}$ such that $(p,w)\in{G}$ and
and process $q\in{Q}$ such that $(w,q)\in{G}$.
Thus, $(p,q)\in{G^2}$ and $\sharedmemcomm{P}{Q}$.
\end{proof}


\begin{lemma}
\label{lem:fmandmneqfnonpartion}
For every $n >4$, there is a shared-memory graph $G$,
such that $\fmandm < \fnonpartionable$ in the uniform m\&m model.
\end{lemma}

\begin{proof}
We show a shared-memory graph $G$ and a set of $n-\fnonpartionable$
processes $P$ that represents at most $\leq \lfloor{n/2}\rfloor$ processes.
The graph $G$ is the star graph over $n$ vertices,
and has edges $\{(p_1,p_2)\}\cup{\{(p_2,p_i):3\leq{i}\leq{n})\}}$. 
(See \figref{ce2}, for $n = 5$.)
Clearly,
there is a path of length $\leq 2$ between every pair of vertices,
and therefore, $G^2$ is a full graph.
From the definition of $\fG$, $\fG=n-1$.
Let $P=\{p_1\}$, note that $|P|=1=n-\fG$.
$P$ represents exactly $|S_{p_1}| = 2 \leq \lfloor{n/2}\rfloor$ processes.
\end{proof}

\lemref{fmandmleqfnonpartion} and \lemref{fmandmneqfnonpartion} imply 
that requiring at least $n-\fmandm$ nonfaulty processes is strictly
stronger than requiring $n-\fnonpartionable$ nonfaulty processes.
Therefore, the HBO algorithm does not have optimal resilience.
Intuitively this happens since HBO does not utilize all the
shared-memory connections that are embodied in $G^2$.
Thus, our algorithm (Section~\ref{sec:consensus}),
has better resilience than HBO, which we show is optimal,
in addition to having polynomial complexity.


Aguilera et al.~\cite{Aguilera:PODC:2018} also present a lower bound
on the number of failures any consensus algorithm can tolerate in
the m\&m model.
To state their bound, consider a graph $G=(V,E)$,
and let $B$, $S$ and $T$ be a partition of $V$.
$(B,S,T)$ is an \emph{SM-cut} in $G$ if $B$ can be partitioned into
two disjoint sets $B_1$ and $B_2$, such that
for every $b_1\in B_1$, $b_2\in B_2$, $s\in S$ and $t\in T$,
we have that $(s,t),(b_1,t),(b_2,s) \notin E$.

\begin{theorem}[\cite{Aguilera:PODC:2018}]
\label{theorem:aguilera-impossibility}
Consensus cannot be solved in the uniform m\&m model
in the presence of $f$ failures if there is a SM-cut $(B,S,T)$
such that $|S|\geq n-f$ and $|T| \geq n-f$.
\end{theorem}

Although the resilience of HBO is not optimal,
we show that this lower bound on resilience is optimal,
by proving that if a system is $f$-partitionable then
the condition in \theoremref{aguilera-impossibility} holds.
By \theoremref{consensus-impossibility},
these two conditions are equal in the m\&m model.

\begin{theorem}
\label{theorem:cut&partition-eq}
In the uniform m\&m model, if the system is $f$-partitionable then
there is an SM-cut $(B,S,T)$ with $|S| \geq n-f$ and $|T| \geq n-f$.
\end{theorem}

\begin{proof}
Since the system is $f$-partitionable,
there are two disjoint sets of processes $S$ and $T$,
each of size $n-f$, such that $\sharedmemcomm[not]{S}{T}$.
Since $\fnonpartionable=\fG$,
$S\cap T=\emptyset$ and there is no edge $(u,v)$ in $G^2$,
where $u\in S$ and $v\in T$.
Let $B'_1=\{u:(u,v)\in E,v\in S\}\setminus{S}$ and
$B_2=\{u:(u,v)\in E, v\in T\}\setminus{T}$.
Obviously, $B'_1\cap S=\emptyset$ and $B_2\cap T=\emptyset$.
$B'_1\cap T = \emptyset$ and $B_2\cap S = \emptyset$, since
there is no edge in $G$ between a vertex in $S$ and a vertex in $T$.
Assume, towards a contradiction, that $B'_1\cap B_2\neq \emptyset$.
Hence, there is a vertex $u$ such that $u\in B'_1\cap B_2$.
Since $u\in B'_1$ then there is some vertex $s\in S$ such that
$(s,u)$ is part of $G$.
Similarly, there is some vertex $t\in T$ such that $(u,t)$ is in $E$.
This implies that the graph $G$ contains the path  $s-u-t$,
contradicting the fact that there is no edge between $S$ and $T$ in $G^2$.
This also implies that there is no edge between a vertex in $T$
and a vertex in $B'_1$, or between a vertex in $S$ and a vertex in $B_2$.
Let $B_1=B'_1\cup(V\setminus(S\cup T \cup B_2))$,
as we showed that $S$, $T$, $B'_1$ and $B_2$ are disjoint
then $S$, $T$, $B_1$ and $B_2$ are a disjoint partition of $V$.
For every $b_1 \in B_1$, $b_2 \in B_2$, $s \in S$ and $t \in T$,
$(s,t)$ ,$(b_1,t)$ and $(b_2,s)$ are not in $E$.
Hence, $(B, S, T)$ is a SM-cut, when $B=B_1\cup B_2$.
\end{proof}

%% file: figs/mm-counter1.pdf_t
\begin{picture}(0,0)%
\includegraphics{mm-counter1.pdf}%
\end{picture}%
\setlength{\unitlength}{4144sp}%
\begingroup\makeatletter\ifx\SetFigFont\undefined%
\gdef\SetFigFont#1#2#3#4#5{%
  \reset@font\fontsize{#1}{#2pt}%
  \fontfamily{#3}\fontseries{#4}\fontshape{#5}%
  \selectfont}%
\fi\endgroup%
\begin{picture}(3434,2356)(7951,-5255)
\put(10036,-3166){\makebox(0,0)[lb]{\smash{{\SetFigFont{20}{24.0}{\rmdefault}{\mddefault}{\updefault}{\color[rgb]{0,0,0}$p_2$}%
}}}}
\put(9046,-3166){\makebox(0,0)[lb]{\smash{{\SetFigFont{20}{24.0}{\rmdefault}{\mddefault}{\updefault}{\color[rgb]{0,0,0}$p_8$}%
}}}}
\put(7966,-3166){\makebox(0,0)[lb]{\smash{{\SetFigFont{20}{24.0}{\rmdefault}{\mddefault}{\updefault}{\color[rgb]{0,0,0}$p_9$}%
}}}}
\put(9856,-4156){\makebox(0,0)[lb]{\smash{{\SetFigFont{20}{24.0}{\rmdefault}{\mddefault}{\updefault}{\color[rgb]{0,0,0}$p_1$}%
}}}}
\put(8011,-5146){\makebox(0,0)[lb]{\smash{{\SetFigFont{20}{24.0}{\rmdefault}{\mddefault}{\updefault}{\color[rgb]{0,0,0}$p_5$}%
}}}}
\put(9091,-5146){\makebox(0,0)[lb]{\smash{{\SetFigFont{20}{24.0}{\rmdefault}{\mddefault}{\updefault}{\color[rgb]{0,0,0}$p_4$}%
}}}}
\put(10081,-5146){\makebox(0,0)[lb]{\smash{{\SetFigFont{20}{24.0}{\rmdefault}{\mddefault}{\updefault}{\color[rgb]{0,0,0}$p_6$}%
}}}}
\put(11161,-5146){\makebox(0,0)[lb]{\smash{{\SetFigFont{20}{24.0}{\rmdefault}{\mddefault}{\updefault}{\color[rgb]{0,0,0}$p_7$}%
}}}}
\put(11116,-3166){\makebox(0,0)[lb]{\smash{{\SetFigFont{20}{24.0}{\rmdefault}{\mddefault}{\updefault}{\color[rgb]{0,0,0}$p_3$}%
}}}}
\end{picture}%

%% file: figs/mm-counter2.pdf_t
\begin{picture}(0,0)%
\includegraphics{mm-counter2.pdf}%
\end{picture}%
\setlength{\unitlength}{4144sp}%
\begingroup\makeatletter\ifx\SetFigFont\undefined%
\gdef\SetFigFont#1#2#3#4#5{%
  \reset@font\fontsize{#1}{#2pt}%
  \fontfamily{#3}\fontseries{#4}\fontshape{#5}%
  \selectfont}%
\fi\endgroup%
\begin{picture}(3434,2266)(8941,-4850)
\put(8956,-3661){\makebox(0,0)[lb]{\smash{{\SetFigFont{20}{24.0}{\rmdefault}{\mddefault}{\updefault}{\color[rgb]{0,0,0}$p_1$}%
}}}}
\put(10441,-3661){\makebox(0,0)[lb]{\smash{{\SetFigFont{20}{24.0}{\rmdefault}{\mddefault}{\updefault}{\color[rgb]{0,0,0}$p_2$}%
}}}}
\put(11836,-2851){\makebox(0,0)[lb]{\smash{{\SetFigFont{20}{24.0}{\rmdefault}{\mddefault}{\updefault}{\color[rgb]{0,0,0}$p_3$}%
}}}}
\put(12151,-3706){\makebox(0,0)[lb]{\smash{{\SetFigFont{20}{24.0}{\rmdefault}{\mddefault}{\updefault}{\color[rgb]{0,0,0}$p_4$}%
}}}}
\put(11881,-4516){\makebox(0,0)[lb]{\smash{{\SetFigFont{20}{24.0}{\rmdefault}{\mddefault}{\updefault}{\color[rgb]{0,0,0}$p_5$}%
}}}}
\end{picture}%

%% file: cluster-based.tex
\section{The Cluster-Based Model}
\label{sec:cluster}

In the hybrid, \emph{cluster-based} model of~\cite{ImbsR13,
RaynalCao:Clusterbased:ICDCS:2019},
processes are partitioned into $m$, $1 \leq m \leq n$,
non-empty and disjoint subsets $P_1, \ldots, P_m$, called \emph{clusters}.
Each cluster has an associated shared memory; only processes of this
cluster can (atomically) read from and write to this shared memory.
The set of processes in the cluster of $p$ is denoted $\textit{cluster}[p]$.
As in the m\&m model,
there are no access restrictions on the shared memory.
Hence, $|R_p|=|W_p|=1$ for every process $p$,
and therefore, $\regsno=n$ and $\readno = O(n^2)$.

In the cluster-based model, $\sharedmemcomm{p}{q}$ if and only if
$p$ and $q$ are in the same cluster.

If $\sharedmemcomm{p}{q}$ and $\sharedmemcomm{q}{w}$,
for some processes $p$, $q$ and $w$, then $p$ and $q$ are in the same cluster
and $q$ and $w$ are in the same cluster.
Since clusters are disjoint, it follows that
$p$ and $w$ are in the same cluster, implying that

\begin{observation}
\label{obs:clustertransitive}
In the cluster-based model,
$\leftrightarrow$ is transitive.
\end{observation}

\begin{definition}\label{def:imp}
$\fc$ is the largest integer $f$ such that for all sets of processes
$P$ and $P'$, each of size $(n-f)$, either $P \cap P' \neq \emptyset$
or some cluster contains a process in $P$ and a process in $P'$.
\end{definition}

\begin{observation}
\label{obs:fnonpartitionisfc}
In the cluster-based model $\fnonpartionable=\fc$.
\end{observation}

\begin{lemma}\label{lem:fmajorityleqfnonpartion}
In the cluster-based model, $\fnonpartionable = \fmajority$.
\end{lemma}

\begin{proof}
Since $\fnonpartionable \leq \fmajority$,
by Lemma~\ref{lem:fnonpartionleqfmajority},
we only need to show that $\fmajority \leq \fnonpartionable$.

Consider two sets of processes $P$ and $Q$,
both of size $n-\fmajority$,
then $|{\communicateset{P}}|>\lfloor{n/2}\rfloor$ and
$|{\communicateset{Q}}|>\lfloor{n/2}\rfloor$.
Therefore, there is a process
$w \in \communicateset{P} \cap \communicateset{Q}$.
Since $w\in{\communicateset{P}}$,
there is a process $p\in{P}$ such that $\sharedmemcomm{p}{w}$,
and since $w\in{\communicateset{Q}}$,
there is a process $q\in{Q}$ such that $\sharedmemcomm{q}{w}$.
Since $\leftrightarrow$ is symmetric and
by \obsref{clustertransitive}, 
it is also transitive, it follows that $\sharedmemcomm{p}{q}$
and thus, $\sharedmemcomm{P}{Q}$.
\end{proof}

\begin{lemma} \label{lem:represent-(n-f)}
In the cluster-based model, for every two sets of processes, $P$ and $Q$,
and $f\leq\fnonpartionable$,
if $|\communicateset{P}|\geq n-f$ and $|\communicateset{Q}|\geq n-f$
then $\sharedmemcomm{P}{Q}$.
\end{lemma}

\begin{proof}
By the assumptions, $|\communicateset{P}|\geq n-\fnonpartionable$
and $|\communicateset{Q}|\geq n-\fnonpartionable$.
By definition,
$\sharedmemcomm{\communicateset{P}}{\communicateset{Q}}$,
implying that there are two processes $p\in\communicateset{P}$
and $q\in\communicateset{Q}$ such that $\sharedmemcomm{p}{q}$.
In the cluster-based model, this means that processes $p$ and $q$
are in some common $P_i$.
Since $p\in \communicateset{P}$,
there is some process $p'\in P$ such that $\sharedmemcomm{p'}{p}$,
Since $q\in \communicateset{Q}$,
there is some process $q'\in Q$ such that $\sharedmemcomm{q'}{q}$.
Since $\leftrightarrow$ is symmetric and transitive,
By \obsref{clustertransitive},  
$p'$ and $q'$ are also in $P_i$,
and hence, $\sharedmemcomm{p'}{q'}$.
\end{proof}

Raynal and Cao~\cite{RaynalCao:Clusterbased:ICDCS:2019} present
two randomized consensus algorithms for the cluster-based model.
One is also based on Ben Or's algorithm~\cite{BenOr1983},
using local coins, and the other is based on an external common coin
(whose implementation is left unspecified).
These algorithms terminate in an execution if there are distinct
clusters whose total size is (strictly) larger than $n/2$,
each containing at least one nonfaulty process.
Clearly, if $f\leq\fmajority$, this condition holds for every execution
with at most $f$ failures.
\lemref{fnonpartionleqfmajority} shows that 
$\fnonpartionable\leq \fmajority$,
the condition holds if there are at most $f\leq\fnonpartionable$ failures.
\lemref{fmajorityleqfnonpartion} implies that these two definitions
are equivalent by proving that $\fnonpartionable=\fmajority$.
This means that the maximum resilience guaranteeing that every two sets
of nonfaulty processes can communicate is equal to the one guaranteeing
that every set of nonfaulty processes
can communicate with a majority of the processes.

In the cluster-based model,
if a process $p \in P_i$ does not crash then all other processes
receive the information from all the processes of $P_i$,
as if none of them crashed.
For this reason, we say that $p$ represents all processes in $P_i$
(note that this definition is different than \defref{represents}).
If a process $q$ receives messages from processes representing $k$
clusters $P_1, \ldots, P_k$,
such that $|P_1| + \cdots + |P_k| > n/2$,
then it has received information from a majority of the processes.
This observation does not change the resilience threshold, i.e.,
the maximal number of failures that can be tolerated,
but allows to wait for a smaller number of messages,
thereby, making the algorithm execute faster.
Lemma~\ref{lem:represent-(n-f)} proves that every two sets of processes
representing at least $n-\fnonpartionable$ processes can communicate.
Therefore, instead of waiting for a majority of represented processes,
as is done in \cite{RaynalCao:Clusterbased:ICDCS:2019},
it suffices to wait for $n-\fnonpartionable$ represented processes.
\algoref{msg-exc-cluster} shows this improvement to the SWMR register
implementation presented in \algoref{swmr},
by replacing the communication pattern, msg\_exchange().
Since $n-\fnonpartionable\leq \lfloor{n/2}\rfloor+1$,
this means that in some cases it suffices to wait for fewer than
a majority of represented processes.

\begin{algorithm}[t]
	\caption{msg\_exchange procedure in the cluster-based model}	
	\label{alg:msg-exc-cluster}
	\begin{algorithmic}[1]
	        \makeatletter\setcounter{ALG@line}{0}\makeatother
			\item[] \blank{-.7cm} \textbf{msg\_exchange\textit{$\langle$m, seq, val$\rangle$}: returns \textit{set of responses}}
			\State send \textit{$\langle$m, seq, val$\rangle$} to all processes
		\State \textit{responses = $\emptyset$} ; \textit{represented = $\emptyset$}
		\Repeat
		\State wait to receive a message $m$ of the form
                    $\langle$\textit{Ack-m, seq, -}$\rangle$ from process $q$
			\State \textit{represented} = \textit{represented} $\cup$ \textit{cluster}$[q]$
		\State \textit{responses} = \textit{responses} $\cup$ \{$m$\}
		\Until $\textit{|represented|}\geq n-f$
		\State return\textit{(responses)}
		\end{algorithmic}
\end{algorithm}

This is not the case in the m\&m model.
For example, in the graph of \figref{ce1}, 
$n=9$ and $\fnonpartionable=6$.
For $P=\{p_7,p_9\}$, $\communicateset{P}=\{p_1,p_6,p_7,p_8,p_9\}$,
and for $Q=\{p_3,p_5\}$, $\communicateset{Q}=\{p_1,p_2,p_3,p_4,p_5\}$,
so $|\communicateset{P}| = |\communicateset{Q}| = 5 > n/2$,
but $\sharedmemcomm[not]{P}{Q}$.
Therefore, even though the system is not $f$-partitionable, and
\lemref{fnonpartionleqfmajority} guarantees that 
the set of non-faulty processes can communicate with a majority of the processes,
it does not suffice to wait for more than $n/2$ represented processes.
(Recall that by \defref{represents}, a process represents all 
the processes it can communicate with using shared memory)

%% file: discussion.tex
\section{Discussion}

This paper studies the optimal resilience for various problems in
mixed models.
Our approach builds on simulating a SWMR register,
which allows to investigate the resilience of many problems,
like implementing MWMR registers and atomic snapshots,
or solving randomized consensus, approximate agreement and renaming.
Prior consensus algorithms for mixed models~\cite{Aguilera:PODC:2018,
RaynalCao:Clusterbased:ICDCS:2019} start from a pure message-passing
algorithm and then try to exploit the added power of shared memory.
In contrast, we start with a shared-memory consensus algorithm
and systematically simulate it in the mixed model.
This simplifies the algorithms and improves their complexity,
while still achieving optimal resilience.

It would be interesting to investigate additional tasks and objects.
An interesting example is \emph{$k$-set consensus}~\cite{Chaudhuri93},
in which processes must decide on at most $k$ different values.
This is trivial for $k=n$ and reduces to consensus, for $k=1$.
For the pure message-passing model, there is a $k$-set consensus
algorithm~\cite{Chaudhuri93}, when the number of failures $f < k$.
This bound is necessary for solving the problem in shared memory
systems~\cite{BorowskyG93,HerlihyS93,SaksZ93}.
Since resilience in a mixed system cannot be better than in the
shared-memory model, it follows that $f < k$ is necessary and sufficient
for any mixed model.
Thus, when $\fnonpartionable<k-1$, a system can be $f$-partitionable and
still offer $f$-resilience for $k$-set consensus.\footnote{
    There is a lower bound of $k >\frac{n-1}{n-f}$ for pure message-passing
    systems, proved using a partitioning argument~\cite{BielyRS2011easy}.
    It might seem that adding shared memory will allow to
    reduce this bound, however, this is not the case, since for the
    relevant ranges of $k$ ($1 < k < n$),
    the bound on the number of failures implied from this bound
    is at least $k$.}


The weakest failure detector needed for implementing a register
in the cluster-based model is strictly weaker than the weakest
failure detector needed in the pure message-passing model~\cite{ImbsR13}.
This aligns with the improved resilience we can achieve in a
mixed model compared to the pure message-passing model.
It is interesting to explore the precise improvement in resilience
achieved with specific failure detectors and other mixed models.

We would also like to study systems where
the message-passing network is not a clique.

\paragraph*{Acknowledgements:}
We thank Vassos Hadzilacos, Xing Hu, Sam Toueg 
and the anonymous reviewers for helpful comments.
This research was supported by ISF grant 380/18.

%% file: main.bbl
\begin{thebibliography}{10}

\bibitem{Gen-Z}
{Gen-Z draft core specification}.
\newblock
  \url{https://genzconsortium.org/specification/gen-z-core-specification-1-1-draft/}.
\newblock Accessed: 2020-08-26.

\bibitem{InfiniBand}
{InfiniBand}.
\newblock \url{https://www.infinibandta.org/about-infiniband/}.
\newblock Accessed: 2020-08-26.

\bibitem{iWARP}
{iWARP}.
\newblock \url{https://en.wikipedia.org/wiki/IWARP}.
\newblock Accessed: 2020-08-26.

\bibitem{RDMA}
{RDMA over converged ethernet}.
\newblock \url{https://en.wikipedia.org/wiki/RDMA_over_Converged_Ethernet}.
\newblock Accessed: 2020-08-26.

\bibitem{Aguilera:PODC:2018}
Marcos~K. Aguilera, Naama Ben-David, Irina Calciu, Rachid Guerraoui, Erez
  Petrank, and Sam Toueg.
\newblock Passing messages while sharing memory.
\newblock In {\em PODC}, page 51–60, 2018.

\bibitem{AspnesH1990}
James Aspnes and Maurice Herlihy.
\newblock Fast randomized consensus using shared memory.
\newblock {\em Journal of Algorithms}, 11(3):441--461, September 1990.

\bibitem{AttiyaBD1995}
Hagit Attiya, Amotz Bar-Noy, and Danny Dolev.
\newblock Sharing memory robustly in message-passing systems.
\newblock {\em Journal of the ACM}, 42(1):124--142, January 1995.

\bibitem{AttiyaBDPR90}
Hagit Attiya, Amotz Bar{-}Noy, Danny Dolev, David Peleg, and R{\"{u}}diger
  Reischuk.
\newblock Renaming in an asynchronous environment.
\newblock {\em Journal of the ACM}, 37(3):524--548, 1990.

\bibitem{AttiyaE2019}
Hagit Attiya and Constantin Enea.
\newblock Putting strong linearizability in context: Preserving hyperproperties
  in programsthat use concurrent objects.
\newblock In {\em DISC}, pages 2:1--2:17, 2019.

\bibitem{AttiyaF01}
Hagit Attiya and Arie Fouren.
\newblock Adaptive and efficient algorithms for lattice agreement and renaming.
\newblock {\em {SIAM} J. Comput.}, 31(2):642--664, 2001.
\newblock \href {https://doi.org/10.1137/S0097539700366000}
  {\path{doi:10.1137/S0097539700366000}}.

\bibitem{AttiyaFG2002}
Hagit Attiya, Arie Fouren, and Eli Gafni.
\newblock An adaptive collect algorithm with applications.
\newblock {\em Distributed Computing}, 15(2):87--96, 2002.

\bibitem{AttiyaLS94}
Hagit Attiya, Nancy~A. Lynch, and Nir Shavit.
\newblock Are wait-free algorithms fast?
\newblock {\em J. {ACM}}, 41(4):725--763, 1994.
\newblock \href {https://doi.org/10.1145/179812.179902}
  {\path{doi:10.1145/179812.179902}}.

\bibitem{AttiyaR98}
Hagit Attiya and Ophir Rachman.
\newblock Atomic snapshots in {$O(n \log n)$} operations.
\newblock {\em {SIAM} J. Comput.}, 27(2):319--340, 1998.

\bibitem{BarNoyD1993}
Amotz Bar{-}Noy and Danny Dolev.
\newblock A partial equivalence between shared-memory and message-passing in an
  asynchronous fail-stop distributed environment.
\newblock {\em Math. Syst. Theory}, 26(1):21--39, 1993.
\newblock \href {https://doi.org/10.1007/BF01187073}
  {\path{doi:10.1007/BF01187073}}.

\bibitem{BenOr1983}
Michael Ben-Or.
\newblock Another advantage of free choice: Completely asynchronous agreement
  protocols.
\newblock In {\em PODC}, pages 27--30, 1983.

\bibitem{BielyRS2011easy}
Martin Biely, Peter Robinson, and Ulrich Schmid.
\newblock Easy impossibility proofs for $k$-set agreement in message passing
  systems.
\newblock In {\em OPODIS}, pages 299--312, 2011.

\bibitem{BorowskyG93}
Elizabeth Borowsky and Eli Gafni.
\newblock Generalized {FLP} impossibility result for t-resilient asynchronous
  computations.
\newblock In {\em STOC}, pages 91--100, 1993.

\bibitem{BrachaT1985}
Gabriel Bracha and Sam Toueg.
\newblock Asynchronous consensus and broadcast protocols.
\newblock {\em Journal of the ACM}, 32(4):824--840, 1985.

\bibitem{Chaudhuri93}
Soma Chaudhuri.
\newblock More choices allow more faults: Set consensus problems in totally
  asynchronous systems.
\newblock {\em Inf. Comput.}, 105(1):132--158, 1993.
\newblock \href {https://doi.org/10.1006/inco.1993.1043}
  {\path{doi:10.1006/inco.1993.1043}}.

\bibitem{Feller57}
Willliam Feller.
\newblock {\em An Introduction to Probability Theory and its Applications,
  volume 1}.
\newblock John Wiley \& Sons, 1957.

\bibitem{FichHS98}
Faith~E. Fich, Maurice Herlihy, and Nir Shavit.
\newblock On the space complexity of randomized synchronization.
\newblock {\em Journal of the ACM}, 45(5):843--862, 1998.
\newblock \href {https://doi.org/10.1145/290179.290183}
  {\path{doi:10.1145/290179.290183}}.

\bibitem{FLP}
Michael~J. Fischer, Nancy~A. Lynch, and Mike Paterson.
\newblock Impossibility of distributed consensus with one faulty process.
\newblock {\em Journal of the ACM}, 32(2):374--382, 1985.
\newblock \href {https://doi.org/10.1145/3149.214121}
  {\path{doi:10.1145/3149.214121}}.

\bibitem{GolabHW2011}
Wojciech Golab, Lisa Higham, and Philipp Woelfel.
\newblock Linearizable implementations do not suffice for randomized
  distributed computation.
\newblock In {\em STOC}, page 373–382, 2011.

\bibitem{Hadzilacos+:M&M:OPODIS:2019}
Vassos Hadzilacos, Xing Hu, and Sam Toueg.
\newblock Optimal register construction in m{\&}m systems.
\newblock In {\em {OPODIS}}, pages 28:1--28:16, 2019.

\bibitem{HadzilacosHT2020arxiv}
Vassos Hadzilacos, Xing Hu, and Sam Toueg.
\newblock Optimal register construction in m{\&}m systems (version 3).
\newblock {\em CoRR}, abs/1906.00298, 2020.
\newblock URL: \url{http://arxiv.org/abs/1906.00298}.

\bibitem{HadzilacosHT2020regular}
Vassos Hadzilacos, Xing Hu, and Sam Toueg.
\newblock Randomized consensus with regular registers.
\newblock {\em CoRR}, abs/2006.06771, 2020.
\newblock URL: \url{http://arxiv.org/abs/2006.06771}.

\bibitem{HerlihyS93}
Maurice Herlihy and Nir Shavit.
\newblock The asynchronous computability theorem for t-resilient tasks.
\newblock In {\em STOC}, pages 111--120, 1993.

\bibitem{HerlihyW1990}
Maurice~P. Herlihy and Jeannette~M. Wing.
\newblock Linearizability: {A} correctness condition for concurrent objects.
\newblock {\em ACM Transactions on Programming Languages and Systems},
  12(3):463--492, July 1990.

\bibitem{ImbsR13}
Damien Imbs and Michel Raynal.
\newblock The weakest failure detector to implement a register in asynchronous
  systems with hybrid communication.
\newblock {\em Theor. Comput. Sci.}, 512:130--142, 2013.
\newblock \href {https://doi.org/10.1016/j.tcs.2012.06.030}
  {\path{doi:10.1016/j.tcs.2012.06.030}}.

\bibitem{Lamport86}
Leslie Lamport.
\newblock On interprocess communication---part {I}: Basic formalism.
\newblock {\em Distributed Computing}, pages 77--85, 1986.

\bibitem{Lim+:ISCA:2009}
Kevin Lim, Jichuan Chang, Trevor Mudge, Parthasarathy Ranganathan, Steven~K.
  Reinhardt, and Thomas~F. Wenisch.
\newblock Disaggregated memory for expansion and sharing in blade servers.
\newblock {\em SIGARCH Comput. Archit. News}, 37(3):267–278, June 2009.
\newblock \href {https://doi.org/10.1145/1555815.1555789}
  {\path{doi:10.1145/1555815.1555789}}.

\bibitem{RaynalCao:Clusterbased:ICDCS:2019}
Michel Raynal and Jiannong Cao.
\newblock One for all and all for one: Scalable consensus in a hybrid
  communication model.
\newblock In {\em {ICDCS}}, pages 464--471, 2019.

\bibitem{SaksZ93}
Michael~E. Saks and Fotios Zaharoglou.
\newblock Wait-free $k$-set agreement is impossible: the topology of public
  knowledge.
\newblock In {\em STOC}, pages 101--110, 1993.

\bibitem{VitanyiA86}
Paul M.~B. Vit{\'{a}}nyi and Baruch Awerbuch.
\newblock Atomic shared register access by asynchronous hardware.
\newblock In {\em FOCS}, pages 233--243, 1986.

\end{thebibliography}
